\documentclass[11pt,letterpaper,reqno]{amsart}

\title{The Hagedorn--Hermite Correspondence}
\author{Tomoki Ohsawa}
\address{Department of Mathematical Sciences, The University of Texas at Dallas, 800 W Campbell Rd, Richardson, TX 75080-3021}
\email{tomoki@utdallas.edu}
\date{\today}



\usepackage{graphicx,mathrsfs,paralist}
\usepackage{mathtools}
\usepackage[margin=1in, marginpar=.5in]{geometry}

\usepackage{tikz-cd}

\usepackage[numbers,sort&compress]{natbib}

\usepackage[colorlinks=false]{hyperref}

\theoremstyle{plain}
\newtheorem{theorem}{Theorem}[section]
\newtheorem{corollary}[theorem]{Corollary}
\newtheorem{lemma}[theorem]{Lemma}
\newtheorem{proposition}[theorem]{Proposition}

\theoremstyle{definition}
\newtheorem{definition}[theorem]{Definition}
\newtheorem{example}[theorem]{Example}

\theoremstyle{remark}
\newtheorem{remark}[theorem]{Remark}

\def\pd#1#2{\dfrac{\partial #1}{\partial #2}}

\def\parentheses#1{{\left(#1\right)}}
\def\brackets#1{{\left[#1\right]}}
\def\braces#1{{\left\{#1\right\}}}

\def\Span{\mathop{\mathrm{span}}\nolimits} 

\def\norm#1{{\left\|#1\right\|}}
\def\abs#1{{\left|#1\right|}}

\def\R{\mathbb{R}}
\def\C{\mathbb{C}}
\def\N{\mathbb{N}}
\def\defeq{\mathrel{\mathop:}=}

\def\setdef#1#2{{\left\{ #1 \ |\ #2 \right\}}}
\def\ip#1#2{{\left\langle#1,#2\right\rangle}}
\def\exval#1{{\left\langle#1\right\rangle}}
\def\texval#1{\langle#1\rangle}
\def\diag{\operatorname{diag}}
\newcommand{\id}{\operatorname{id}}
\renewcommand{\Re}{\operatorname{Re}}
\renewcommand{\Im}{\operatorname{Im}}

\def\GL{\mathsf{GL}}
\def\SO{\mathsf{SO}}

\def\Sp{\mathsf{Sp}}
\def\FSp{\mathsf{FSp}}
\def\Mp{\mathsf{Mp}}
\def\U{\mathsf{U}}
\def\SO{\mathsf{SO}}
\def\Orth{\mathsf{O}}
\def\Mat{\mathsf{M}}
\def\Sym{\mathsf{Sym}}

\newenvironment{tbmatrix}{\left[\begin{smallmatrix}}{\end{smallmatrix}\right]}

\def\rmi{{\rm i}}

\begin{document}

\footskip=.6in

\begin{abstract}
  We investigate the relationship between the semiclassical wave packets of Hagedorn and the Hermite functions by establishing a relationship between their ladder operators.
  This Hagedorn--Hermite correspondence provides a unified view as well as simple proofs of some essential results on the Hagedorn wave packets.
  Particularly, we show that Hagedorn's ladder operators are a natural set of ladder operators obtained from the position and momentum operators using the symplectic group.
  This construction reveals an algebraic structure of the Hagedorn wave packets, and explains the relative simplicity of Hagedorn's parametrization compared to the rather intricate construction of the generalized squeezed states.
  We apply our formulation to show the existence of minimal uncertainty products for the Hagedorn wave packets, generalizing Hagedorn's one-dimensional result to multi-dimensions.
  The Hagedorn--Hermite correspondence also leads to an alternative derivation of the generating function for the Hagedorn wave packets based on the generating function for the Hermite functions.
  This result, in turn, reveals the relationship between the Hagedorn polynomials and the Hermite polynomials.
\end{abstract}

\maketitle

\section{Introduction}
\subsection{The Hagedorn Wave Packets and Generalized Squeezed States}
The Hagedorn wave packets $\{ \varphi^{\hbar}_{n} \}_{n \in \N_{0}^{d}} \subset L^{2}(\R^{d})$ with $\N_{0} \defeq \N \cup \{0\}$ are a set of wave functions with parameters, and have the following remarkable properties (see \citet{Ha1980, Ha1981, Hagedorn1985, Ha1998} and also \citet{Ro2007}): (i)~They are an orthonormal basis for $L^{2}(\R^{d})$ with associated ladder operators. (ii)~Each wave packet $\varphi^{\hbar}_{n}$ is an exact solution to the Schr\"odinger equation with quadratic Hamiltonians when the parameters evolve in time according to a certain set of ordinary differential equations. (iii)~By taking a certain linear combination of a finite subset of $\{ \varphi^{\hbar}_{n} \}_{n \in \N_{0}^{d}}$, one may construct an approximate solution---with an error of $O(\hbar^{N/2})$ for any $N \in \N$---to the Schr\"odinger equation with non-quadratic potentials with some regularity.

It goes without saying that these results give significant insights into solutions of the Schr\"odinger equation in the semiclassical regime $\hbar \ll 1$.
In recent years, the Hagedorn wave packets have been also implemented in numerical computations for the semiclassical Schr\"odinger equation; see, e.g., \citet{FaGrLu2009} and \citet{GrHa2014}.

Many of these theoretical and numerical studies take advantage of the key properties of the Hagedorn wave packets.
As one can see in the series of works of \citet{Ha1980, Ha1981, Hagedorn1985, Ha1998}, the Hagedorn wave packets share many properties with the Hermite functions, most notably the ladder operators discovered in \cite{Ha1998}, which are very useful in simplifying calculations and proofs involving the Hagedorn wave packets.

The relationship between Hagedorn wave packets and the Hermite functions is understood to some extent, but it is rather that those similarities are discovered by inspection case by case: Most of such properties of the Hagedorn wave packets have been proved by generalizing the proofs of the corresponding properties of Hermite functions.
However, these proofs tend to be cumbersome because the Hagedorn wave packets are significantly more complicated than the Hermite functions.

We note that the Hagedorn wave packets are known to be essentially equivalent to the so-called generalized squeezed states (see, e.g., \citet{CoRo2012} and references therein).
While the correspondence between the generalized squeezed states and the Hermite functions is well known, the construction and parametrization of the generalized squeezed states is rather intricate and less explicit, making them less amenable to applications.
The advantage of Hagedorn's approach is the simple and explicit expressions of the ladder operators, and the resulting explicit construction of the wave packets.
These explicit expressions are particularly beneficial in applications such as numerical implementations as mentioned above.

\subsection{Main Results}
Our main motivation is to establish the relationship between Hagedorn and Hermite.
More specifically, we reveal the exact correspondence between the ladder operators of Hagedorn and Hermite (Theorem~\ref{thm:ladder_operators} and Section~\ref{ssec:Hagedorn_ladder_operators}).
This leads to the correspondence between the Hagedorn wave packets and the Hermite functions (Theorem~\ref{thm:Hagedorn-Hermite}).
Furthermore, many properties of the Hagedorn wave packets follow naturally from the corresponding properties of the Hermite functions by exploiting this correspondence.

Our results are complementary to those in the generalized squeezed states literature.
The correspondence of the form in Theorem~\ref{thm:Hagedorn-Hermite} is well known for the generalized squeezed states.
Our contribution is to reveal the algebraic structure behind the Hagedorn--Hermite correspondence and elucidate the relative simplicity of Hagedorn's parametrization by stressing the role played by the symplectic group $\Sp(2d,\R)$.
This is in contrast to the very involved definition and parametrization of the metaplectic operator in the generalized squeezed states literature; see, e.g., \cite[Chapter~3]{CoRo2012}.

Particularly, our result reveals the so-called {\em symplectic covariance} (see, e.g., \citet{Go2011}) of the Hagedorn wave packets.
Symplectic covariance is particularly helpful in simplifying calculations involving metaplectic operators because it essentially turns those calculations involving metaplectic operators into matrix multiplications by the corresponding symplectic matrices.
In fact, it turns out that many of the known results regarding the Hagedorn wave packets turn out to be simple corollaries of some forms of symplectic covariance.

\subsection{Outline}
Section~\ref{sec:CCR_ladder_ops} starts off with a brief review of the set of operators that are written as linear combinations of the standard position and momentum operators.
The main result of this section, Theorem~\ref{thm:ladder_operators}, shows the necessary and sufficient condition for such a set of operators to be ladder operators.
We stress the role of the symplectic group $\Sp(2d,\R)$ in the construction and parametrization of the ladder operators.
The reader who is not familiar with the Heisenberg--Weyl and metaplectic operators may consult the brief review of them in Appendix~\ref{sec:HW_and_Mp} before Section~\ref{ssec:Transformations} as they play a critical role throughout the paper.

Section~\ref{sec:Hagedorn_wave_packets} applies the results from Section~\ref{sec:CCR_ladder_ops} to the setting of the Hagedorn wave packets.
In fact, it turns out that the ladder operators characterized in Theorem~\ref{thm:ladder_operators} (see also Section~\ref{ssec:Hagedorn_ladder_operators}) in terms of the symplectic group $\Sp(2d,\R)$ are essentially those of \citet{Ha1998}.
What follows immediately from it is the relationship between the ladder operators for Hagedorn and Hermite (Proposition~\ref{prop:As-as}); we note that \citet[Proposition~6]{LaTr2014} obtained an essentially the same result in a rather intricate manner.
This result is exploited to build a bridge between the ladders of the Hagedorn wave packets and the Hermite functions (Theorem~\ref{thm:Hagedorn-Hermite}).
We also prove symplectic covariance of the ladder operators and Hagedorn wave packets in this section.
These results yield some of the fundamental and essential results on the Hagedorn wave packets as simple corollaries.

In Section~\ref{sec:minimal_uncertainty}, we apply the approach developed in Section~\ref{sec:Hagedorn_wave_packets} to prove the existence of uncertainty products of the Hagedorn wave packets; this is a multi-dimensional generalization of the one-dimensional result of \citet{Ha2013}.

In Section~\ref{sec:generating_functions}, we obtain the generating function for the Hagedorn wave packets and those polynomials appearing in them (called the {\em Hagedorn polynomials} in this paper) again exploiting the results from Section~\ref{sec:Hagedorn_wave_packets}.
Such a generating function is obtained by \citet{DiKeTr2017} and \citet{Ha2015} using the recurrence relations and by direct calculations, respectively.
Our approach is different from them in the sense that the generating function for the Hagedorn wave packets is obtained directly from that for the Hermite functions; particularly, this is done in a manner that exactly parallels the Hagedorn--Hermite correspondence obtained in Theorem~\ref{thm:Hagedorn-Hermite}.
In other words, we reveal a simple relationship between the generating functions of Hagedorn and Hermite.
The Hagedorn--Hermite correspondence in terms of generating functions in turn yields (Corollary~\ref{cor:Hagedorn-Hermite_polynomials}) the relationship between the Hagedorn and Hermite polynomials as well.

Appendix~\ref{sec:HW_and_Mp} gives a quick review of the Heisenberg--Weyl and metaplectic operators, and Appendix~\ref{sec:Hermite} gives a summary of some known facts on the Hermite functions and Hermite polynomials.
The main purpose of these appendices is to set our notation as well as to include some key results to refer to in the main body in an effort to make the paper as self-contained as possible.

\section{Canonical Commutation Relations and Ladder Operators}
\label{sec:CCR_ladder_ops}
\subsection{Linear Transformations of Position \& Momentum Operators}
Let $\hat{z} = (\hat{x}, \hat{p})$ be the standard position and momentum operators on $L^{2}(\R^{d})$, i.e.,
\begin{equation}
  \label{eq:xhat_phat}
  \hat{x}_{j} f(x) \defeq x_{j} f(x),
  \qquad
  \hat{p}_{j} f(x) \defeq -\rmi\hbar \pd{}{x_{j}}f(x).
\end{equation}
One may then take the Schwartz space $\mathscr{S}(\R^{d})$ as the domain of any of these operators.
We now follow \citet[Section~7.3]{Go2011} (see also \citet[Section~7.1]{Go2006}) to first define the set of linear operators on the Hilbert space $L^{2}(\R^{d})$ that can be written as linear combinations of the $2d$ operators $\hat{z} = (\hat{x}, \hat{p})$ over the complex numbers\footnote{The original definition by \citet{Go2006,Go2011} is over real numbers $\R$.} $\C$, i.e.,
\begin{equation*}
  \Span_{\C}(\{\hat{z}\}) \defeq \setdef{ a \cdot \hat{x} + b \cdot \hat{p} = \sum_{j=1}^{d} (a_{j}\hat{x}_{j} + b_{k}\hat{p}_{k}) }{ a, b \in \C^{d} }.
\end{equation*}
Since each of the $2d$ operators $\hat{z} = (\hat{x}, \hat{p})$ is defined on the Schwartz space $\mathscr{S}(\R^{d})$, so is any linear combination; hence we may take $\mathscr{S}(\R^{d})$ as the domain of any operator in $\Span_{\C}(\{\hat{z}\})$.

Clearly $\Span_{\C}(\{\hat{z}\})$ is a vector space isomorphic to $\C^{2d}$.
We may then define the linear isomorphism
\begin{equation*}
  \varrho(\,\cdot\,; \hat{z}) \colon \C^{2d} \to \Span_{\C}(\{\hat{z}\});
  \qquad
  c = (c_{1}, c_{2}) \mapsto c_{1} \cdot \hat{p} - c_{2} \cdot \hat{x},
\end{equation*}
or equivalently, using the matrix
\begin{equation*}
  J \defeq
  \begin{bmatrix}
    0 & I_{d} \\
    -I_{d} & 0
  \end{bmatrix},
\end{equation*}
we have
\begin{equation}
  \label{eq:varrho}
  \varrho(c;\hat{z}) = c^{T} J \hat{z}.
\end{equation}
Alternatively, one may regard, with a slight abuse of notation, $\hat{z}$ as a vector in the symplectic vector space $T^{*}\R^{d} \cong \R^{2d}$ with the standard symplectic form $\Omega$ defined by $\Omega(v, w) = v^{T} J w$ for $v, w \in \R^{2d}$, and can rewrite the above expression in the following succinct form:
\begin{equation*}
  \varrho(c;\hat{z}) = \Omega(c, \hat{z}).
\end{equation*}

\begin{remark}
  \label{rem:why_J_in_varrho}
  Having $J$ in the definition~\eqref{eq:varrho} is crucial in making sure that $\varrho$ has the ``symplectic covariance'' property as we shall see in \eqref{eq:varrho-covariance} of Section~\ref{ssec:Transformations} below.
\end{remark}

Now let us extend this idea further to define a set of $2d$ operators, each of which belongs to $\Span_{\C}(\{\hat{z}\})$.
Specifically, let $\Mat_{2d}(\C)$ be the set of complex $2d \times 2d$ matrices, and define the homomorphism
\begin{subequations}
  \label{eq:rho}
  \begin{equation}
    \rho(\,\cdot\,; \hat{z}) \colon \Mat_{2d}(\C) \to \Span_{\C}(\{\hat{z}\})^{2d}
    = \underbrace{ \Span_{\C}(\{\hat{z}\}) \oplus \dots \oplus \Span_{\C}(\{\hat{z}\}) }_{\text{$2d$ copies}}
  \end{equation}
  as follows:
  \begin{equation}
    \rho(\mathcal{X}; \hat{z}) \defeq \mathcal{X}^{T} J \hat{z}.
  \end{equation}
  In particular, writing $\mathcal{X} =
  \begin{tbmatrix}
    \mathsf{A} & \mathsf{B} \smallskip\\
    \mathsf{C} & \mathsf{D}
  \end{tbmatrix}$ with $\mathsf{A}, \mathsf{B}, \mathsf{C}, \mathsf{D} \in \Mat_{d}(\C)$, we have
  \begin{equation}
    \rho\parentheses{
      \begin{bmatrix}
        \mathsf{A} & \mathsf{B} \smallskip\\
        \mathsf{C} & \mathsf{D}
      \end{bmatrix}
    } = 
    \begin{bmatrix}
      -\mathsf{C}^{T} & \mathsf{A}^{T} \smallskip\\
      -\mathsf{D}^{T} & \mathsf{B}^{T}
    \end{bmatrix}
    \begin{bmatrix}
      \hat{x} \smallskip\\
      \hat{p}
    \end{bmatrix}
    = 
    \begin{bmatrix}
      -\mathsf{C}^{T} \hat{x} + \mathsf{A}^{T} \hat{p} \smallskip\\
      -\mathsf{D}^{T} \hat{x} + \mathsf{B}^{T} \hat{p}
    \end{bmatrix}.
  \end{equation}
\end{subequations}
We mention in passing that a similar idea of defining such complex transformation is discussed in \citet{Wo1974}.

It turns out that it is convenient to group the resulting $2d$ operators into two---one consisting of the first $d$ operators and the other the rest of them---and so we may also define
\begin{equation*}
  \begin{array}{c}
    \rho^{\flat}(\,\cdot\,; \hat{z})\colon \Mat_{2d}(\C) \to \Span_{\C}(\{\hat{z}\})^{d};
    \qquad
    \rho^{\flat}(\mathcal{X}; \hat{z}) \defeq -\mathsf{C}^{T} \hat{x} + \mathsf{A}^{T} \hat{p},
    \smallskip\\
    \rho^{\sharp}(\,\cdot\,; \hat{z})\colon \Mat_{2d}(\C) \to \Span_{\C}(\{\hat{z}\})^{d};
    \qquad
    \rho^{\sharp}(\mathcal{X}; \hat{z}) \defeq -\mathsf{D}^{T} \hat{x} + \mathsf{B}^{T} \hat{p}.
  \end{array}
\end{equation*}
The motivation behind this grouping is that, as well shall see below, we will later characterize $\rho^{\flat}(\mathcal{X}; \hat{z})$ and $\rho^{\sharp}(\mathcal{X}; \hat{z})$ as lowering and raising operators, respectively, under a certain assumption on the matrix $\mathcal{X} \in \Mat_{2d}(\C)$.

\subsection{Symplectic Group $\Sp(2d,\R)$ and Ladder Operators}
So far we did not impose any additional assumptions on the matrix $\mathcal{X} \in \Mat_{2d}(\C)$.
In this subsection, we show the necessary and sufficient condition for the matrix $\mathcal{X}$ so that the set of operators $\rho(\mathcal{X}; \hat{z})$ defines ladder operators on $\mathscr{S}(\R^{d})$.

We first mention an auxiliary result when restricting $\mathcal{X}$ to $\Mat_{2d}(\R) \subset \Mat_{2d}(\C)$.
Let $\Sp(2d,\R)$ be the symplectic group of degree $2d$, i.e.,
\begin{subequations}
  \begin{equation*}
    \Sp(2d,\R) \defeq \setdef{ S \in \Mat_{2d}(\R) }{ S^{T} J S = J },
  \end{equation*}
  or equivalently, written as block matrices consisting of $d \times d$ submatrices, i.e., $S =
  \begin{tbmatrix}
    A & B \smallskip\\
    C & D
  \end{tbmatrix}$ with $A, B, C, D \in \Mat_{d}(\R)$, 
  \begin{equation}
    \label{def:Sp2dR-block}
    \Sp(2d,\R) \defeq
    \setdef{
      \begin{bmatrix}
        A & B \\
        C & D
      \end{bmatrix}
      \in \Mat_{2d}(\R)
    }{ A^{T}C = C^{T}A,\, B^{T}D = D^{T}B,\, A^{T}D - C^{T}B = I_{d} }.
  \end{equation}
\end{subequations}
Then it is straightforward to see the following:
\begin{proposition}[{\citet[Section~4.1]{Li1986}; see also \citet[Appendix~B]{Wo1974}}]
  \label{prop:Littlejohn}
  Let $\mathcal{X} \in \Mat_{2d}(\R)$ and $\hat{z}$ be $2d$ symmetric operators on $L^{2}(\R^{d})$ that satisfy the canonical commutation relations on $\mathscr{S}(\R^{d})$.
  Then the set of $2d$ operators $\rho(\mathcal{X}; \hat{z})$ defined in \eqref{eq:rho} are also symmetric operators on $L^{2}(\R^{d})$ that satisfy the canonical commutation relations on $\mathscr{S}(\R^{d})$ if and only if $\mathcal{X} \in \Sp(2d,\R)$.
\end{proposition}

\begin{proof}
  It is clear, for any $j \in \{1, \dots, 2d\}$, that $\rho(\mathcal{X}; \hat{z})_{j}$ is symmetric because $\mathcal{X} \in \Mat_{2d}(\R)$ and also that $\rho(\mathcal{X}; \hat{z})_{j} f \in \mathscr{S}(\R^{d})$ for any $f \in \mathscr{S}(\R^{d})$.
  Straightforward calculations yield
  \begin{align*}
    [\rho_{j}(\mathcal{X}; \hat{z}), \rho_{k}(\mathcal{X}; \hat{z})]
    &= [ (\mathcal{X}^{T} J \hat{z})_{j}, (\mathcal{X}^{T} J \hat{z})_{k} ]
    \\
    &= \rmi \hbar\, (\mathcal{X}^{T} J \mathcal{X})_{jk},
  \end{align*}
  But then $\mathcal{X}^{T}J\mathcal{X} = J$ if and only if $\mathcal{X} \in \Sp(2d,\R)$.
\end{proof}

\begin{remark}
  \citet[Section~4.1]{Li1986} does not have $J$ in the definition of $\rho$, but as far as this result is concerned, it is equivalent to the above.
  Having $J$ is important for us to maintain symplectic covariance of $\rho$ as alluded above; see Proposition~\ref{prop:rho-covariance} below.
  \citet[Appendix~B]{Wo1974} discusses more or less an equivalent result in terms of a {\em complex} matrix $\mathcal{X} \in \Mat_{2d}(\C)$, which is also related to Theorem~\ref{thm:ladder_operators} below.
\end{remark}

The goal of this subsection is to come up with a condition for $\mathcal{X} \in \Mat_{2d}(\C)$ so that the set of $2d$ operators $\rho(\mathcal{X}; \hat{z})$ gives ladder operators.
Let us first define what we mean by ladder operators:
\begin{definition}[Ladder operators on $\mathscr{S}(\R^{d})$]
  \label{def:ladder_operators}
  Let $\mathcal{X} \in \Mat_{2d}(\C)$ and $\hat{z}$ be $2d$ symmetric operators on $L^{2}(\R^{d})$ that satisfy the canonical commutation relations on $\mathscr{S}(\R^{d})$.
  We say that the $2d$ operators $\rho(\mathcal{X}; \hat{z})$ defined by \eqref{eq:rho} are {\em ladder operators} on $\mathscr{S}(\R^{d})$ if the following conditions are satisfied for any $j, k \in \{1, \dots, 2d\}$:
  \begin{enumerate}[(i)]
  \item $[\rho_{j}(\mathcal{X}; \hat{z}), \rho_{k}(\mathcal{X}; \hat{z})] = J_{jk}$ and
    \smallskip
  \item for any $f \in \mathscr{S}(\R^{d})$,
    \begin{subequations}
      \label{def:adjoint_condition}
      \begin{equation}
        \rho^{\sharp}_{j}(\mathcal{X}; \hat{z})f = \rho^{\flat}_{j}(\mathcal{X}; \hat{z})^{*}f,
      \end{equation}
      i.e., we have
      \begin{equation}
        \text{ $\rho^{\sharp}(\mathcal{X}; \hat{z}) = \rho^{\flat}(\mathcal{X}; \hat{z})^{*}$ on $\mathscr{S}(\R^{d})$ }.
      \end{equation}
    \end{subequations}
  \end{enumerate}
  More specifically, we call $\rho^{\flat}(\mathcal{X}; \hat{z})$ the {\em lowering operators} and $\rho^{\sharp}(\mathcal{X}; \hat{z})$ the {\em raising operators}.
\end{definition}

The ladder operators for the harmonic oscillator is a special example of the above definition.
To see this, let us first define the following set of matrices:
Let us define a unitary matrix $W \in \U(2d)$ by
\begin{equation*}
  W \defeq \frac{1}{\sqrt{2}}
  \begin{bmatrix}
    \rmi I_{d} & -\rmi I_{d} \\
    -I_{d} & -I_{d}
  \end{bmatrix}.
\end{equation*}
This is related to the unitary matrix
\begin{equation}
  \label{eq:mathcalW}
  \mathcal{W} \defeq
  \frac{1}{\sqrt{2}}
  \begin{bmatrix}
    I_{d} & \rmi I_{d} \\
    I_{d} & -\rmi I_{d}
  \end{bmatrix} 
  \in \U(2d)
\end{equation}
defined in \citet[Eq.~(4.12) on p.~174]{Fo1989} in the following way:
\begin{equation}
  \label{eq:W-mathcalW}
  \mathcal{W} = W^{T} J.
\end{equation}
We may also incorporate the small parameter $\hbar$ by defining
\begin{equation}
  \label{eq:W_hbar}
  W_{\hbar} \defeq \frac{1}{\sqrt{\hbar}} W
  = \frac{1}{\sqrt{2\hbar}}
  \begin{bmatrix}
    \rmi I_{d} & -\rmi I_{d} \\
    -I_{d} & -I_{d}
  \end{bmatrix},
\end{equation}
It is easy to see that $W_{\hbar}$ satisfies
\begin{equation*}
  W_{\hbar}^{T} J W_{\hbar} = -\frac{\rmi}{\hbar} J.
\end{equation*}
Similarly, we may define
\begin{equation*}
  \mathcal{W}_{\hbar} \defeq \frac{1}{\sqrt{\hbar}} \mathcal{W}
  = \frac{1}{\sqrt{2\hbar}}
  \begin{bmatrix}
    I_{d} & \rmi I_{d} \\
    I_{d} & -\rmi I_{d}
  \end{bmatrix}
\end{equation*}
so that 
\begin{equation*}
  \mathcal{W}_{\hbar} = W_{\hbar}^{T} J.
\end{equation*}

Now it is easy to see that the ladder operators for the harmonic oscillator are defined in terms of $\rho$ and the above matrices:
\begin{example}[Ladder operators for harmonic oscillator]
  \label{ex:harmonic_oscillator}
  Let $\hat{z} = (\hat{x},\hat{p})$ be the standard position and momentum operators.
  Set $\mathcal{X} = W_{\hbar}$ and define
  \begin{equation}
    \label{def:as}
    \begin{bmatrix}
      \hat{a} \smallskip\\
      \hat{a}^{*}
    \end{bmatrix}
    \defeq
    \begin{bmatrix}
      \rho^{\flat}(W_{\hbar}; \hat{z}) \smallskip\\
      \rho^{\sharp}(W_{\hbar}; \hat{z})
    \end{bmatrix}
    = \rho(W_{\hbar}; \hat{z})
    = W_{\hbar}^{T} J \hat{z}
    = \mathcal{W}_{\hbar}\, \hat{z}.
  \end{equation}
  This yields the ladder operators for the harmonic oscillator:
  \begin{equation*}
    \hat{a} \defeq \rho^{\flat}(W_{\hbar}; \hat{z}) = \frac{1}{\sqrt{2\hbar}} (\hat{x} + \rmi\,\hat{p}),
    \qquad
    \hat{a}^{*} \defeq \rho^{\sharp}(W_{\hbar}; \hat{z}) = \frac{1}{\sqrt{2\hbar}} (\hat{x} - \rmi\,\hat{p}).
  \end{equation*}
  They clearly satisfy the conditions in Definition~\ref{def:ladder_operators} and thus define ladder operators on $\mathscr{S}(\R^{d})$.
\end{example}

It turns out that, for any symplectic matrix $S \in \Sp(2d,\R)$, the set of $2d$ operators $\rho(S W_{\hbar}; \hat{z})$ also defines ladder operators; furthermore, conversely, any set $\rho(\mathcal{X}; \hat{z})$ of ladder operators in the sense of Definition~\ref{def:ladder_operators} can be written as $\rho(S W_{\hbar}; \hat{z})$ with some $S \in \Sp(2d,\R)$:
\begin{theorem}
  \label{thm:ladder_operators}
  Let $\hat{z}$ be $2d$ symmetric operators on $L^{2}(\R^{d})$ that satisfy the canonical commutation relations on $\mathscr{S}(\R^{d})$.
  \begin{enumerate}[(i)]
  \item \label{thm:ladder_operators-i}
    The $2d$ operators $\rho(\mathcal{X}; \hat{z})$ with $\mathcal{X} \in \Mat_{2d}(\C)$ are ladder operators on $\mathscr{S}(\R^{d})$ if and only if $\mathcal{X} = S W_{\hbar}$ with $S \in \Sp(2d,\R)$.
  \item \label{thm:ladder_operators-ii}
    With a symplectic matrix $S = 
    \begin{tbmatrix}
      A & B \smallskip\\
      C & D
    \end{tbmatrix} \in \Sp(2d,\R)$, the ladder operators $\rho(S W_{\hbar}; \hat{z})$ take the form
    \begin{equation}
      \label{eq:ladder_operators}
      \rho(S W_{\hbar}; \hat{z})
      = \frac{1}{\sqrt{\hbar}} \mathcal{W} S^{-1} \hat{z}.
    \end{equation}
    Specifically, the lowering and raising operators are given by
    \begin{subequations}
      \label{eq:lowering_and_raising_operators}
      \begin{align}
        \rho^{\flat}(S W_{\hbar}; \hat{z})  &= -\frac{\rmi}{\sqrt{2\hbar}} \brackets{ (C + \rmi D)^{T} \hat{x} - (A + \rmi B)^{T} \hat{p} }, \\
        \rho^{\sharp}(S W_{\hbar}; \hat{z}) &= \frac{\rmi}{\sqrt{2\hbar}} \brackets{ (C - \rmi D)^{T} \hat{x} - (A - \rmi B)^{T} \hat{p} },
      \end{align}
    \end{subequations}
    respectively.
  \end{enumerate}
\end{theorem}

\begin{proof}
  First recall from the proof of Proposition~\ref{prop:Littlejohn} that
  \begin{equation*}
    [\rho(\mathcal{X}; \hat{z})_{j}, \rho(\mathcal{X}; \hat{z})_{k}]
    = \rmi \hbar\, (\mathcal{X}^{T} J \mathcal{X})_{jk}.
  \end{equation*}
  This implies that
  \begin{equation*}
    [\rho_{j}(\mathcal{X}; \hat{z}), \rho_{k}(\mathcal{X}; \hat{z})] = J_{jk}
    \iff
    \mathcal{X}^{T} J \mathcal{X} = -\frac{\rmi}{\hbar} J.
  \end{equation*}

  Let us first prove the sufficiency in \eqref{thm:ladder_operators-i} and also \eqref{thm:ladder_operators-ii}.
  It is a straightforward calculation to check that the above relationship holds with $\mathcal{X} = W_{\hbar}$, i.e.,
  \begin{equation*}
    W_{\hbar}^{T} J W_{\hbar} = -\frac{\rmi}{\hbar} J
  \end{equation*}
  as mentioned above.
  Therefore, by setting $\mathcal{X} = S W_{\hbar}$ with any $S \in \Sp(2d,\R)$, we have
  \begin{equation*}
    \mathcal{X}^{T} J \mathcal{X} = W_{\hbar}^{T} S^{T} J S W_{\hbar} = W_{\hbar}^{T} J W_{\hbar} = -\frac{\rmi}{\hbar} J.
  \end{equation*}
  Moreover, it is easy to see that $\rho^{\sharp}(S W_{\hbar}; \hat{z}) = \rho^{\flat}(S W_{\hbar}; \hat{z})^{*}$ as well:
  First notice that
  \begin{equation*}
    \rho(S W_{\hbar}; \hat{z})
    = (S W_{\hbar})^{T} J \hat{z}
    = \frac{1}{\sqrt{\hbar}} (S W)^{T} J \hat{z}
    = \frac{1}{\sqrt{\hbar}} \mathcal{W} S^{-1} \hat{z},
  \end{equation*}
  where we used the following equality: Using $S^{T} J S = J \iff S^{T} J = J S^{-1}$ and \eqref{eq:W-mathcalW},
  \begin{equation*}
    (S W)^{T} J = W^{T} S^{T} J = W^{T} J S^{-1} = \mathcal{W} S^{-1}.
  \end{equation*}
  By writing $S =
  \begin{tbmatrix}
    A & B \smallskip\\
    C & D
  \end{tbmatrix}$, we have $S^{-1} = - J S^{T} J = 
  \begin{tbmatrix}
    D^{T} & -B^{T} \smallskip\\
    -C^{T} & A^{T}
  \end{tbmatrix}$ and so
  \begin{equation*}
    \mathcal{W} S^{-1} = \frac{1}{\sqrt{2}}
    \begin{bmatrix}
      D^{T} - \rmi C^{T} & -B^{T} + \rmi A^{T} \smallskip\\
      D^{T} + \rmi C^{T} & -B^{T} - \rmi A^{T}
     \end{bmatrix}
     = \frac{\rmi}{\sqrt{2}}
    \begin{bmatrix}
      - (C + \rmi D)^{T} & (A + \rmi B)^{T} \smallskip\\
       (C - \rmi D)^{T} & -(A - \rmi B)^{T}
     \end{bmatrix}.
  \end{equation*}
  Therefore,
  \begin{equation*}
    \rho(S W_{\hbar}; \hat{z}) =
    \begin{bmatrix}
      \rho^{\flat}(S W_{\hbar}; \hat{z}) \smallskip\\
      \rho^{\sharp}(S W_{\hbar}; \hat{z}) 
    \end{bmatrix}
    = \frac{\rmi}{\sqrt{2\hbar}}
    \begin{bmatrix}
      -(C + \rmi D)^{T} & (A + \rmi B)^{T} \smallskip\\
      (C - \rmi D)^{T} & -(A - \rmi B)^{T}
    \end{bmatrix}
    \begin{bmatrix}
      \hat{x} \smallskip\\
      \hat{p}
    \end{bmatrix}.
  \end{equation*}
  and so we see that $\rho^{\sharp}(S W_{\hbar}; \hat{z}) = \rho^{\flat}(S W_{\hbar}; \hat{z})^{*}$.
  This proves the sufficiency in \eqref{thm:ladder_operators-i} as well as \eqref{thm:ladder_operators-ii}.

  For the necessity in \eqref{thm:ladder_operators-i}, let us first set $\mathcal{X} = \begin{tbmatrix}
    \mathsf{A} & \mathsf{B} \smallskip\\
    \mathsf{C} & \mathsf{D}
  \end{tbmatrix}$ with $\mathsf{A}, \mathsf{B}, \mathsf{C}, \mathsf{D} \in \Mat_{d}(\C)$.
  Then
  \begin{align*}
    \rho^{\flat}(\mathcal{X}; \hat{z}) &= -\mathsf{C}^{T} \hat{x} + \mathsf{A}^{T} \hat{p}, \\
    \rho^{\sharp}(\mathcal{X}; \hat{z}) &= -\mathsf{D}^{T} \hat{x} + \mathsf{B}^{T} \hat{p},
  \end{align*}
  and so the condition~\eqref{def:adjoint_condition} on the adjoints, i.e., $\rho^{\sharp}(\mathcal{X}; \hat{z}) = \rho^{\flat}(\mathcal{X}; \hat{z})^{*}$ on $\mathscr{S}(\R^{d})$, implies that
  \begin{equation}
    \label{eq:adjoint_condition}
    \mathsf{B} = \overline{\mathsf{A}}
    \quad\text{and}\quad
    \mathsf{D} = \overline{\mathsf{C}}.
  \end{equation}
  Also recall from above that $[\rho_{j}(\mathcal{X}; \hat{z}), \rho_{k}(\mathcal{X}; \hat{z})] = J_{jk}$ is equivalent to $\mathcal{X}^{T} J \mathcal{X} = -\frac{\rmi}{\hbar} J$; but then this in turn is equivalent to the following conditions on the block components:
  \begin{gather}
    \mathsf{A}^{T} \mathsf{C} = \mathsf{C}^{T} \mathsf{A}, \qquad
    \mathsf{B}^{T} \mathsf{D} = \mathsf{D}^{T} \mathsf{B}, \label{eq:ladder_condition-1}\\
    \mathsf{A}^{T} \mathsf{D} - \mathsf{C}^{T} \mathsf{B} = -\frac{\rmi}{\hbar} I_{d}. \label{eq:ladder_condition-2}
  \end{gather}
  The second equation in \eqref{eq:ladder_condition-1} is equivalent to the first one due to \eqref{eq:adjoint_condition}, and so is redundant here.
  Now, writing $\mathsf{A} = A_{1} + \rmi A_{2}$ and $\mathsf{C} = C_{1} + \rmi C_{2}$ with $A_{j}, C_{j} \in \Mat_{d}(\R)$ for $j = 1, 2$, the first equation in \eqref{eq:ladder_condition-1} is equivalent to
  \begin{gather}
    A_{1}^{T} C_{1} - C_{1}^{T} A_{1} = A_{2}^{T} C_{2} - C_{2}^{T} A_{2}, \label{eq:A_C-1} \\
    A_{1}^{T} C_{2} - C_{1}^{T} A_{2} = C_{2}^{T} A_{1} - A_{2}^{T} C_{1}, \label{eq:A_C-2}
  \end{gather}
  whereas \eqref{eq:ladder_condition-2} combined with \eqref{eq:adjoint_condition} gives
  \begin{gather}
    A_{1}^{T} C_{1} - C_{1}^{T} A_{1} = -(A_{2}^{T} C_{2} - C_{2}^{T} A_{2}), \label{eq:A_C-3} \\
    (A_{1}^{T} C_{2} - C_{1}^{T} A_{2}) + (C_{2}^{T} A_{1} - A_{2}^{T} C_{1}) = \frac{1}{\hbar} I_{d}. \label{eq:A_C-4}
  \end{gather}
  Now \eqref{eq:A_C-1} and \eqref{eq:A_C-3} together imply
  \begin{equation}
    \label{eq:A_C-symplecticity-1}
    A_{1}^{T} C_{1} = C_{1}^{T} A_{1}
    \quad\text{and}\quad
    A_{2}^{T} C_{2} = C_{2}^{T} A_{2},
  \end{equation}
  whereas \eqref{eq:A_C-2} and \eqref{eq:A_C-4} give
  \begin{equation}
    \label{eq:A_C-symplecticity-2}
    2\hbar( C_{2}^{T} A_{1} - A_{2}^{T} C_{1} ) = I_{d}.
  \end{equation}
  To conclude the proof, notice that \eqref{eq:adjoint_condition} implies that the matrix $\mathcal{X} \in \Mat_{2d}(\C)$ takes the following form:
  \begin{equation*}
    \mathcal{X} =
    \begin{bmatrix}
      \mathsf{A} & \mathsf{B} \smallskip\\
      \mathsf{C} & \mathsf{D}
    \end{bmatrix}
     =
    \begin{bmatrix}
      \mathsf{A} & \overline{\mathsf{A}} \smallskip\\
      \mathsf{C} & \overline{\mathsf{C}}
    \end{bmatrix}
    = \sqrt{2\hbar}
    \begin{bmatrix}
      A_{2} & -A_{1} \smallskip\\
      C_{2} & -C_{1}
    \end{bmatrix}
    \frac{1}{\sqrt{2\hbar}}
    \begin{bmatrix}
      \rmi I_{d} & -\rmi I_{d} \smallskip\\
      -I_{d} & -I_{d}
    \end{bmatrix}
    = S W_{\hbar},
  \end{equation*}
  where we set $S = \sqrt{2\hbar}
  \begin{tbmatrix}
    A_{2} & -A_{1} \smallskip\\
    C_{2} & -C_{1}
  \end{tbmatrix}$ and $W_{\hbar}$ is defined in \eqref{eq:W_hbar}; but then the conditions \eqref{eq:A_C-symplecticity-1} and \eqref{eq:A_C-symplecticity-2} imply that $S \in \Sp(2d,\R)$ in view of \eqref{def:Sp2dR-block}.
\end{proof}

\subsection{Transformations under the Heisenberg--Weyl and Metaplectic Operators}
\label{ssec:Transformations}
It turns out that an operator of the form $\varrho(c;\hat{z})$ defined in \eqref{eq:varrho} transforms rather nicely under the Heisenberg--Weyl and metaplectic operators (see Appendix~\ref{sec:HW_and_Mp} for a brief review of these operators), and as a result, so does $\rho(\mathcal{X}; \hat{z})$ defined in \eqref{eq:rho}.

Recall (see Section~\ref{ssec:HW_operator}) that the Heisenberg--Weyl operator~\eqref{eq:That} transforms the standard position and momentum operators $\hat{z}$ as follows:
\begin{equation*}
  \widehat{T}_{z}\,\hat{z}\,\widehat{T}_{z}^{*} = \hat{z} - z.
\end{equation*}
Then it follows easily from the definition~\eqref{eq:varrho} of $\varrho$ that, for any $c \in \C^{2d}$ and $j \in \{1, \dots, 2d\}$, we also have
\begin{equation*}
  \widehat{T}_{z}\,\varrho(c; \hat{z})_{j}\,\widehat{T}_{z}^{*} = \rho(c; \hat{z} - z)_{j}.
\end{equation*}
A similar property holds with the metaplectic operators as well.
Namely, we have the following symplectic covariance property (see, e.g., \citet[Section~7.3.1]{Go2011}) alluded in Remark~\ref{rem:why_J_in_varrho}:
For any metaplectic operator $\widehat{S} \in \Mp(2d,\R)$ with $S = \pi_{\Mp}(\widehat{S}) \in \Sp(2d,\R)$,
\begin{equation}
  \label{eq:varrho-covariance}
  \widehat{S} \varrho(c; \hat{z}) \widehat{S}^{*} = \varrho(S c; \hat{z}).
\end{equation}

\begin{remark}
  As mentioned in Remark~\ref{rem:why_J_in_varrho}, having $J$ in the definition~\eqref{eq:varrho} of $\varrho$ is critical because defining, e.g., $\tilde{\varrho}(c; \hat{z}) = c^{T} \hat{z}$ without the matrix $J$ results in violating the symplectic covariance.
  This is easy to see by checking that, e.g., $\widehat{M}_{L} \tilde{\varrho}(c; \hat{z}) \widehat{M}_{L}^{*} \neq \tilde{\varrho}(M_{L} c; \hat{z})$, where $\widehat{M}_{L}$ is one of the generators of $\Mp(2d,\R)$ defined in \eqref{eq:Mhat2} of Section~\ref{ssec:Mp}.
\end{remark}

These properties of transformations of operators $\varrho(c; \hat{z})$ can be easily extended to those $\rho(\mathcal{X}; \hat{z})$ defined in \eqref{eq:rho}:

\begin{proposition}
  \label{prop:rho-covariance}
  Let $z \in T^{*}\R^{d}$ and $\widehat{T}_{z}$ be the corresponding Heisenberg operator~\eqref{eq:That}, and $\widehat{S} \in \Mp(2d,\R)$ be a metaplectic operator corresponding to $S \in \Sp(2d,\R)$, i.e., $\pi_{\Mp}(\widehat{S}) = S$.
  Then, for any $\mathcal{X} \in \Mat_{2d}(\C)$,
  \begin{equation*}
    \widehat{T}_{z}\, \rho(\mathcal{X}; \hat{z})\, \widehat{T}_{z}^{*} = \rho(\mathcal{X}; \hat{z} - z),
    \qquad
    \widehat{S}\, \rho(\mathcal{X}; \hat{z})\, \widehat{S}^{*} = \rho(S\mathcal{X}; \hat{z}),
  \end{equation*}
  i.e., for any $j \in \{1, \dots, d\}$, we have
  \begin{equation*}
    \widehat{T}_{z}\, \rho_{j}(\mathcal{X}; \hat{z})\, \widehat{T}_{z}^{*} = \rho_{j}(\mathcal{X}; \hat{z} - z),
    \qquad
    \widehat{S}\, \rho_{j}(\mathcal{X}; \hat{z})\, \widehat{S}^{*} = \rho_{j}(S\mathcal{X}; \hat{z}).
  \end{equation*}
\end{proposition}

\begin{proof}
  Follows easily from the above transformation formulas for $\varrho$ since, writing $\mathcal{X}$ in terms of column vectors, i.e., $\mathcal{X} = [\mathcal{X}_{*1}\,|\, \dots \,|\, \mathcal{X}_{*d}]$, we have
  \begin{equation*}
    \rho(\mathcal{X}; \hat{z}) =
    \begin{bmatrix}
      \varrho(\mathcal{X}_{*1}; \hat{z}) \\
      \vdots \\
      \varrho(\mathcal{X}_{*d}; \hat{z})
    \end{bmatrix},
  \end{equation*}
  i.e., $\rho_{j}(\mathcal{X}; \hat{z}) = \varrho(\mathcal{X}_{*j}; \hat{z})$.
  Therefore, we have
  \begin{align*}
    \widehat{S}\,\rho_{j}(\mathcal{X}; \hat{z})\,\widehat{S}^{*} &= \widehat{S}\,\varrho(\mathcal{X}_{*j}; \hat{z})\,\widehat{S}^{*} \\
    &= \varrho(S\mathcal{X}_{*j}; \hat{z}) \\
    &= \varrho((S\mathcal{X})_{*j}; \hat{z}) \\
    &= \rho_{j}(S\mathcal{X}; \hat{z}).
  \end{align*}
  A similar calculation yields the desired equality for the Heisenberg--Weyl operator $\widehat{T}_{z}$ as well.
\end{proof}

\section{Ladder Operators and Semiclassical Wave Packets of Hagedorn}
\label{sec:Hagedorn_wave_packets}
\subsection{The Hagedorn Ladder Operators}
\label{ssec:Hagedorn_ladder_operators}
Let $S \in \Sp(2d,\R)$, $\widehat{S} \in \Mp(2d,\R)$ be a corresponding metaplectic operator, $z \defeq (q, p) \in T^{*}\R^{d}$, and $\hat{z} = (\hat{x}, \hat{p})$ be the position and momentum operators~\eqref{eq:xhat_phat}.
We now define a set of operators $(\mathscr{A}(S,z), \mathscr{A}^{*}(S,z))$ by
\begin{subequations}
  \label{def:As}
  \begin{equation}
    \begin{bmatrix}
      \mathscr{A}(S,z) \smallskip\\
      \mathscr{A}^{*}(S,z)
    \end{bmatrix}
    \defeq \widehat{T}_{z}\, \widehat{S}\, \rho(W_{\hbar}; \hat{z})\, \widehat{S}^{*}\, \widehat{T}_{z}^{*},
  \end{equation}
  i.e., for any $j \in \{1, \dots, d\}$,
  \begin{equation*}
    \mathscr{A}_{j}(S,z) \defeq \widehat{T}_{z}\, \widehat{S}\, \rho^{\flat}_{j}(W_{\hbar}; \hat{z})\, \widehat{S}^{*}\, \widehat{T}_{z}^{*},
    \qquad
    \mathscr{A}^{*}_{j}(S,z) \defeq \widehat{T}_{z}\, \widehat{S}\, \rho^{\sharp}_{j}(W_{\hbar}; \hat{z})\, \widehat{S}^{*}\, \widehat{T}_{z}^{*}.
  \end{equation*}
  Note that the ambiguity in the sign of $\widehat{S} \in \Mp(2d,\R)$ (see Appendix~\ref{ssec:Mp}) is immaterial here because the signs cancel out by the conjugation.
  Using Proposition~\ref{prop:rho-covariance} and \eqref{eq:ladder_operators}, we may write
  \begin{equation}
    \begin{bmatrix}
      \mathscr{A}(S,z) \smallskip\\
      \mathscr{A}^{*}(S,z)
    \end{bmatrix}
    = \rho(S W_{\hbar}; \hat{z} - z) = 
    \begin{bmatrix}
      \rho^{\flat}(S W_{\hbar}; \hat{z} - z) \smallskip\\
      \rho^{\sharp}(S W_{\hbar}; \hat{z} - z)
    \end{bmatrix}
    = \frac{1}{\sqrt{\hbar}} \mathcal{W} S^{-1} (\hat{z} - z).
  \end{equation}
\end{subequations}
Since the symmetric operators $\widehat{z} - z = (\hat{x} - q, \hat{p} - p)$ clearly satisfy the canonical commutation relations on $\mathscr{S}(\R^{d})$, Theorem~\ref{thm:ladder_operators} implies that $( \mathscr{A}(S,z), \mathscr{A}^{*}(S,z) )$ define ladder operators on $\mathscr{S}(\R^{d})$ with $\mathscr{A}(S,z) = \rho^{\flat}(S W_{\hbar}; \hat{z} - z)$ being the lowering operators and $\mathscr{A}^{*}(S,z) = \rho^{\sharp}(S W_{\hbar}; \hat{z} - z)$ being the raising operators; hence for any $j, k \in \{1, \dots, d\}$,
\begin{equation}
  \label{eq:As-commutators}
   [\mathscr{A}_{j}(S,z), \mathscr{A}_{k}(S,z)] = 0,
   \qquad
   [\mathscr{A}_{j}^{*}(S,z), \mathscr{A}^{*}_{k}(S,z)] = 0,
   \qquad
  [\mathscr{A}_{j}(S,z), \mathscr{A}^{*}_{k}(S,z)] = \delta_{jk}.
\end{equation}

We now employ the following parametrization of $S \in \Sp(2d,\R)$ of \citet[Section~V.1]{Lu2008}:
\begin{align}
  \label{def:Sp2d-block-Lubich}
  \Sp(2d,\R)
  &= \setdef{
    \begin{bmatrix}
      \Re Q & \Im Q \smallskip\\
      \Re P & \Im P
    \end{bmatrix} \in \Mat_{2d}(\R)
              }
              {
              \begin{array}{c}
                Q, P \in \Mat_{d}(\C),\ Q^{T}P - P^{T}Q = 0, \smallskip\\
                Q^{*}P - P^{*}Q = 2\rmi I_{d}
              \end{array}
              },
\end{align}
that is, we set
\begin{equation}
  \label{eq:S_in_QP}
  S \defeq 
  \begin{bmatrix}
    \Re Q & \Im Q \smallskip\\
    \Re P & \Im P
  \end{bmatrix} \in \Sp(2d,\R)
\end{equation}
in \eqref{def:As}.
Then, using the expressions~\eqref{eq:lowering_and_raising_operators} from Theorem~\ref{thm:ladder_operators}, we have
\begin{align*}
  \mathscr{A}(S,z) &= -\frac{\rmi}{\sqrt{2\hbar}} \brackets{ P^{T}(\hat{x} - q) - Q^{T}(\hat{p} - p) }, \\
  \mathscr{A}^{*}(S,z) &= \frac{\rmi}{\sqrt{2\hbar}} \brackets{ P^{*}(\hat{x} - q) - Q^{*}(\hat{p} - p) }.
\end{align*}
We recognize them as the ladder operators of \citet{Ha1998} (\citeauthor{Ha1998} uses parameters $A, B \in \Mat_{d}(\C)$, which are related to $Q$ and $P$ as $A = Q$ and $B = -\rmi P$; see also \citet[Section~V.2]{Lu2008}).

As we shall see later in Section~\ref{ssec:Hagedorn-Hermite}, the normalized ``ground state'' $\varphi^{\hbar}_{0}(S,z;\,\cdot\,)$ of the Hagedorn wave packet contains the $d \times d$ complex matrix $P Q^{-1}$ in its quadratic term inside the exponential (see \eqref{eq:varphi_0} below).
In fact, one can show that if $S = \begin{tbmatrix}
  \Re Q & \Im Q \smallskip\\
  \Re P & \Im P
\end{tbmatrix} \in \Sp(2d,\R)$ then $P Q^{-1}$ is an element in the Siegel upper half space
\begin{equation*}
  \Sigma_{d} \defeq 
  \setdef{ \mathcal{A} + {\rm i}\mathcal{B} \in \Mat_{d}(\mathbb{C}) }{ \mathcal{A}, \mathcal{B} \in \Mat_{d}(\R),\, \mathcal{A}^{T} = \mathcal{A},\, \mathcal{B}^{T} = \mathcal{B},\, \mathcal{B} > 0 },
\end{equation*}
i.e., the set of symmetric $d \times d$ complex matrices (symmetric in the real sense) with positive-definite imaginary parts; this guarantees that $\varphi^{\hbar}_{0}(S,z;\,\cdot\,)$ is an element in $L^{2}(\R^{d})$ (again see \eqref{eq:varphi_0} below).

\begin{remark}
  Geometrically, this is because the Siegel upper half space $\Sigma_{d}$ is identified as the homogeneous space $\Sp(2d,\R)/\U(d)$, where the (transitive) action is defined as the following generalized linear fractional transformation:
  \begin{equation}
    \label{eq:Psi}
    \Psi\colon \Sp(2d,\R) \times \Sigma_{d} \to \Sigma_{d};
    \quad
    \parentheses{
      \begin{bmatrix}
        A & B \\
        C & D
      \end{bmatrix},
      \mathcal{Z}
    }
    \mapsto
    (C + D\mathcal{Z})(A + B\mathcal{Z})^{-1}.
  \end{equation}
  This action gives rise to the following natural quotient map (see \citet{Si1943}, \citet[Section~4.5]{Fo1989}, and \citet[Exercise~2.28 on p.~48]{McSa1999}; see also \citet{Oh2015c}):
  \begin{align*}
    \pi_{\U(d)}\colon & \Sp(2d,\R) \to \Sp(2d,\R)/\U(d) \cong \Sigma_{d};
    \\
    & Y = \begin{bmatrix}
      A & B \\
      C & D
    \end{bmatrix}
    \mapsto
    \Psi_{Y}({\rm i}I_{d}) = (C + {\rm i}D)(A + {\rm i}B)^{-1}.
  \end{align*}
  Therefore, with the parametrization~\eqref{def:Sp2d-block-Lubich} for $S \in \Sp(2d,\R)$, we have
  \begin{equation}
    \label{eq:PQinv}
    \pi_{\U(d)}(S) = 
    \pi_{\U(d)}\parentheses{
      \begin{bmatrix}
        \Re Q & \Im Q \smallskip\\
        \Re P & \Im P
      \end{bmatrix}
    }
    = P Q^{-1} \in \Sigma_{d}.
  \end{equation}
\end{remark}

Now let us go back to the definition~\eqref{def:As} of the Hagedorn ladder operators.
We observed in Example~\ref{ex:harmonic_oscillator} that $\rho(W_{\hbar}; \hat{z})$ gives the ladder operators $( \hat{a}, \hat{a}^{*} )$ for the harmonic oscillator.
This implies the following:
\begin{proposition}
  \label{prop:As-as}
  The ladder operators $( \mathscr{A}(S,z), \mathscr{A}^{*}(S,z) )$ of \citet{Ha1998} are related to those $( \hat{a}, \hat{a}^{*} )$ for the harmonic oscillator (see \eqref{def:as}) as follows:
  \begin{equation}
    \label{eq:As-as}
    \mathscr{A}_{j}(S,z) = \widehat{T}_{z}\, \widehat{S}\, \hat{a}_{j}\, \widehat{S}^{*}\, \widehat{T}_{z}^{*},
    \qquad
    \mathscr{A}^{*}_{j}(S,z) = \widehat{T}_{z}\, \widehat{S}\, \hat{a}^{*}_{j}\, \widehat{S}^{*}\, \widehat{T}_{z}^{*},
  \end{equation}
  that is, the diagrams
  \begin{equation}
    \label{cd:As-as}
    \begin{tikzcd}[column sep=7ex, row sep=6ex, ampersand replacement=\&]
      \mathscr{S}(\R^{d}) \arrow{r}{\widehat{T}_{z}\,\widehat{S}} \arrow{d}[swap]{\hat{a}_{j}} \& \mathscr{S}(\R^{d}) \arrow{d}{\mathscr{A}_{j}(S,z)}
      \\
      \mathscr{S}(\R^{d}) \arrow{r}[swap]{\widehat{T}_{z}\,\widehat{S}} \& \mathscr{S}(\R^{d})
    \end{tikzcd}
    \qquad
    \begin{tikzcd}[column sep=7ex, row sep=6ex, ampersand replacement=\&]
      \mathscr{S}(\R^{d}) \arrow{r}{\widehat{T}_{z}\,\widehat{S}} \& \mathscr{S}(\R^{d})
      \\
      \mathscr{S}(\R^{d}) \arrow{r}[swap]{\widehat{T}_{z}\,\widehat{S}} \arrow{u}{\hat{a}^{*}_{j}} \& \mathscr{S}(\R^{d}) \arrow{u}[swap]{\mathscr{A}^{*}_{j}(S,z)}
    \end{tikzcd}
  \end{equation}
  commute for each $j \in \{1, \dots, d\}$.
\end{proposition}

\begin{proof}
  Follows easily from the definition~\eqref{def:As} of $( \mathscr{A}(S,z), \mathscr{A}^{*}(S,z) )$ and the fact that $( \hat{a}, \hat{a}^{*} )$ are given as $\rho(W_{\hbar}; \hat{z})$ as shown in Example~\ref{ex:harmonic_oscillator}.
\end{proof}

\begin{remark}
  An essentially the same result is obtained by \citet[Proposition~6]{LaTr2014} by intricate calculations involving the squeezing operators; see also, e.g., \citet[Section~3.4]{CoRo2012}.
\end{remark}

\subsection{Symplectic Covariance of the Hagedorn Ladder Operators}
The above characterization of the ladder operators of \citet{Ha1998} leads to the following symplectic covariance property of the ladder operators:
\begin{proposition}[Symplectic covariance of Hagedorn ladder operators]
  \label{prop:A_covariance}
  For any $\widehat{S}_{0} \in \Mp(2d,\R)$ with $S_{0} \defeq \pi_{\Mp}(\widehat{S}_{0}) \in \Sp(2d,\R)$, the ladder operators $( \mathscr{A}(S,z), \mathscr{A}^{*}(S,z) )$ satisfy
  \begin{equation*}
    \widehat{S}_{0}\, \mathscr{A}_{j}(S, z)\, \widehat{S}_{0}^{*} = \mathscr{A}_{j}(S_{0} S, S_{0} z),
    \qquad
    \widehat{S}_{0}\, \mathscr{A}^{*}_{j}(S, z)\, \widehat{S}_{0}^{*} = \mathscr{A}^{*}_{j}(S_{0} S, S_{0} z)
  \end{equation*}
  for each $j \in \{1, \dots, d\}$.
\end{proposition}

\begin{proof}
  First recall from~\eqref{def:As} that 
  \begin{equation*}
    \begin{bmatrix}
      \mathscr{A}(S,z) \smallskip\\
      \mathscr{A}^{*}(S,z)
    \end{bmatrix}
    \defeq \widehat{T}_{z}\, \widehat{S}\, \rho(W_{\hbar}; \hat{z})\, \widehat{S}^{*}\, \widehat{T}_{z}^{*}.
  \end{equation*}
  Then, using the symplectic covariance~\eqref{eq:T-S_covariance} of the Heisenberg--Weyl operator $\widehat{T}_{z}$ and Proposition~\ref{prop:rho-covariance}, we have 
  \begin{align*}
    \widehat{S}_{0}
    \begin{bmatrix}
      \mathscr{A}(S,z) \smallskip\\
      \mathscr{A}^{*}(S,z)
    \end{bmatrix}
    \widehat{S}_{0}^{*}
    &= \widehat{S}_{0}\, \widehat{T}_{z}\, \widehat{S}\, \rho(W_{\hbar}; \hat{z})\, \widehat{S}^{*}\, \widehat{T}_{z}^{*}\, \widehat{S}_{0}^{*} \\
    &= \widehat{T}_{S_{0} z}\, \widehat{S}_{0}\, \widehat{S}\, \rho(W_{\hbar}; \hat{z})\, \widehat{S}^{*}\, \widehat{S}_{0}^{*}\, \widehat{T}_{S_{0} z}^{*} \\
    &= \rho(S_{0} S W_{\hbar}; \hat{z} - S_{0} z) \\
    &= \begin{bmatrix}
      \mathscr{A}(S_{0}S, S_{0}z) \smallskip\\
      \mathscr{A}^{*}(S_{0}S, S_{0}z)
    \end{bmatrix}. \qedhere
  \end{align*}
\end{proof}

This is a generalization of the transformation property of the ladder operators under the conjugation by the Fourier transform in \citet{Ha1998}:
Let $\mathscr{F}_{\hbar}$ be the semiclassical Fourier transform defined as 
\begin{equation}
  \label{eq:mathcalF_hbar}
  \mathscr{F}_{\hbar}\psi(x) = \frac{1}{ (2\pi\hbar)^{d/2} } \int_{\R^{d}} e^{-\frac{\rmi}{\hbar}x \cdot \tilde{x}}\, \psi(\tilde{x})\,d\tilde{x}
\end{equation}
on $\mathscr{S}(\R^{d})$.
Then it is easy to see the following:
\begin{corollary}[{\citet[Eqs.~(2.12) and (2.13)]{Ha1998}}]
  The ladder operators $( \mathscr{A}(S,z), \mathscr{A}^{*}(S,z) )$ satisfy, for each $j \in \{1, \dots, d\}$,
  \begin{equation*}
    \mathscr{F}_{\hbar}\, \mathscr{A}_{j}(S, z)\, \mathscr{F}_{\hbar}^{-1} = \mathscr{A}_{j}(J S, J z),
    \qquad
    \mathscr{F}_{\hbar}\, \mathscr{A}^{*}_{j}(S, z)\, \mathscr{F}_{\hbar}^{-1} = \mathscr{A}^{*}_{j}(J S, J z),
  \end{equation*}
  where $\mathscr{F}_{\hbar}$ is the semiclassical Fourier transform defined in \eqref{eq:mathcalF_hbar}; or equivalently, writing
  \begin{equation*}
    \mathscr{A}(Q, P, q, p) \defeq \mathscr{A}(S, z),
    \qquad
    \mathscr{A}^{*}(Q, P, q, p) \defeq \mathscr{A}^{*}(S, z),
  \end{equation*}
  we have, for each $j \in \{1, \dots, d\}$,
  \begin{equation*}
    \mathscr{F}_{\hbar}\, \mathscr{A}_{j}(Q, P, q, p)\, \mathscr{F}_{\hbar}^{-1} = \mathscr{A}_{j}(P, -Q, p, -q),
    \qquad
    \mathscr{F}_{\hbar}\, \mathscr{A}^{*}_{j}(Q, P, q, p)\, \mathscr{F}_{\hbar}^{-1} = \mathscr{A}^{*}_{j}(P, -Q, p, -q).
  \end{equation*}
\end{corollary}

\begin{remark}
  \label{eq:QP_versus_AB}
  The apparent difference from Eqs.~(2.12) and (2.13) of \citet{Ha1998} by the constant factors $\pm\rmi$ stems from different parametrizations of elements in $S \in \Sp(2d,\R)$.
  Namely, the parameters $(A, B)$ in \citet{Ha1998} correspond to ours (originally due to \citet[Section~V.1]{Lu2008}) as $A = Q$ and $B = -\rmi P$.
  This implies that what corresponds to the transformation $(Q, P) \mapsto (P, -Q)$ is $(A, B) \mapsto (\rmi B, \rmi A)$ in Hagedorn's parametrization, and the imaginary unit $\rmi$ is pulled out of the expressions to appear as the constant factors $\pm\rmi$ in Eqs.~(2.12) and (2.13) of \cite{Ha1998}.
\end{remark}

\begin{proof}
  Using the identity $\mathscr{F}_{\hbar} = \rmi^{d/2} \widehat{J}$ (see \eqref{eq:Jhat} for the definition of $\widehat{J} \in \Mp(2d,\R)$), we have
  \begin{equation*}
    \mathscr{F}_{\hbar}\, \mathscr{A}_{j}(S, z)\, \mathscr{F}_{\hbar}^{-1}
    = \widehat{J}\, \mathscr{A}_{j}(S, z)\, \widehat{J}^{*},
    \qquad
    \mathscr{F}_{\hbar}\, \mathscr{A}^{*}_{j}(S, z)\, \mathscr{F}_{\hbar}^{-1}
    = \widehat{J}\, \mathscr{A}^{*}_{j}(S, z)\, \widehat{J}^{*}
  \end{equation*}
  for each $j \in \{1, \dots, d\}$.
  Then, setting $\widehat{S}_{0} = \widehat{J}$ in the above proposition, we have $S_{0} = J$ (see \eqref{eq:pi_Mp_of_generators}) here and so the above proposition gives
  \begin{equation*}
    \mathscr{F}_{\hbar}\, \mathscr{A}_{j}(S, z)\, \mathscr{F}_{\hbar}^{-1}
    = \mathscr{A}_{j}(J S, J z)
    \qquad
    \mathscr{F}_{\hbar}\, \mathscr{A}^{*}_{j}(S, z)\, \mathscr{F}_{\hbar}^{-1}
    = \mathscr{A}^{*}_{j}(J S, J z).
  \end{equation*}
  Since 
  \begin{equation*}
    J S =
    \begin{bmatrix}
      \Re P & \Im P \smallskip\\
      -\Re Q & -\Im Q
    \end{bmatrix},
    \qquad
    J z =
    \begin{bmatrix}
      p \smallskip\\
      -q
    \end{bmatrix},
  \end{equation*}
  the maps $S \mapsto J S$ and $z \mapsto J z$ correspond to $(Q, P) \mapsto (P, -Q)$ and $(q, p) \mapsto (p, -q)$, respectively.
\end{proof}

\subsection{The Hagedorn Wave Packets and the Hermite Functions}
\label{ssec:Hagedorn-Hermite}
Recall that the ground state $\psi^{\hbar}_{0} \in L^{2}(\R^{d})$ for the harmonic oscillator may be defined as
\begin{equation*}
  \hat{a}_{j} \psi^{\hbar}_{0} = 0
  \text{ for any $j \in \{1, \dots, d\}$}
  \quad\text{and}\quad
  \bigl\| \psi^{\hbar}_{0} \bigr\| = 1,
\end{equation*}
where $\norm{\,\cdot\,}$ is the $L^{2}$ norm.
It is easy to find (modulo the phase factor)
\begin{equation}
  \label{eq:psi_0}
  \psi^{\hbar}_{0}(x) = \frac{1}{(\pi\hbar)^{d/4}}\,\exp\parentheses{-\frac{x^{2}}{2\hbar}}.
\end{equation}
Likewise, one may define the ``ground state'' of the Hagedorn ladder operators as
\begin{equation*}
  \mathscr{A}_{j}(S,z)\, \varphi^{\hbar}_{0}(S,z;\,\cdot\,) = 0
  \text{ for any $j \in \{1, \dots, d\}$}
  \quad\text{and}\quad
  \bigl\| \varphi^{\hbar}_{0}(S,z;\,\cdot\,) \bigr\| = 1,
\end{equation*}
but then it is easy to see from \eqref{eq:As-as} of Proposition~\ref{prop:As-as} that, for any $j \in \{1, \dots, d\}$,
\begin{equation*}
  \mathscr{A}_{j}(S,z)\, \widehat{T}_{z}\, \widehat{S}\, \psi^{\hbar}_{0} = \widehat{T}_{z}\, \widehat{S}\, \hat{a}_{j}\, \psi^{\hbar}_{0} = 0,
\end{equation*}
and also that $\bigl\| \widehat{T}_{z}\, \widehat{S}\, \psi^{\hbar}_{0} \bigr\| = 1$ because both $\widehat{T}_{z}$ and $\widehat{S}$ are unitary.
So we would like to define the ground state $\varphi^{\hbar}_{0}(S,z;\,\cdot\,)$ as
\begin{equation}
  \label{eq:Hagedorn-Hermite_0}
  \varphi^{\hbar}_{0}(S,z;\,\cdot\,) \defeq e^{-\frac{\rmi}{2\hbar} p\cdot q}\, \widehat{T}_{z}\, \widehat{S}\, \psi^{\hbar}_{0}.
\end{equation}
We put the extra phase factor $e^{-\frac{\rmi}{2\hbar} p\cdot q}$ because one can show (see, e.g., \citet[Section~7.2]{Li1986} and \citet[Theorem~4.65 on p.~202]{Fo1989}) that this definition coincides with that of \citet{Ha1980,Ha1998}:
\begin{equation}
  \label{eq:varphi_0}
  \varphi^{\hbar}_{0}(S,z; x) = \frac{(\det Q)^{-1/2}}{(\pi\hbar)^{d/4}} \exp\braces{ \frac{{\rm i}}{\hbar}\brackets{ \frac{1}{2}(x - q)^{T}P Q^{-1}(x - q) + p \cdot (x - q) } }.
\end{equation}

\begin{remark}
  Strictly speaking, the above expression represents two functions that differ by the sign, depending on how one takes the branch cut in defining the square root $(\det Q)^{1/2}$.
  The same goes with many of those functions to follow that are defined to be parametrized by $S \in \Sp(2d,\R)$ and contain factors like $(\det Q)^{-1/2}$ in their expressions.
  They are in fact parametrized by $\widehat{S} \in \Mp(2d,\R)$ and hence is double-valued.
  Nevertheless, we ostensibly parametrize those functions by $S \in \Sp(2d,\R)$ or $(Q,P)$ and let the square root term take care of the ambiguity in the sign.
\end{remark}

\citet{Ha1998} generated an orthonormal basis $\{ \varphi^{\hbar}_{n}(S,z;\,\cdot\,) \}_{n \in \N_{0}^{d}}$ for $L^{2}(\R^{d})$ by applying the raising operator recursively just as is done with the Hermite functions in \eqref{eq:psi_n-ladders}, i.e., for any multi-index $n = (n_{1}, \dots, n_{d}) \in \N_{0}^{d}$ and $j \in \{1, \dots, d\}$,
\begin{equation}
  \label{eq:varphi_n-raised}
  \varphi^{\hbar}_{n + e_{j}}(S,z;\,\cdot\,) \defeq \frac{1}{\sqrt{n_{j} + 1}}\,\mathscr{A}^{*}_{j} \varphi^{\hbar}_{n}(S,z;\,\cdot\,).
\end{equation}
where $e_{j}$ is the unit vector in $\R^{d}$ whose $j$-th entry is 1.
One can also show (see \citet{Ha1998}) inductively that
\begin{equation}
  \label{eq:varphi_n-lowered}
  \varphi^{\hbar}_{n - e_{j}}(S,z;\,\cdot\,) \defeq \frac{1}{\sqrt{n_{j}}}\,\mathscr{A}_{j} \varphi^{\hbar}_{n}(S,z;\,\cdot\,).
\end{equation}
It is also easy to see that each $\varphi^{\hbar}_{n}(S, z; x)$ is the ground state $\varphi^{\hbar}_{0}(S, z; x)$ multiplied by a polynomial in $x$.
Therefore, for any multi-index $n \in \N_{0}^{d}$, we may define the polynomial
\begin{equation}
  \label{eq:P_n}
  \mathcal{P}^{\hbar}_{n}(S, z; x) \defeq c_{n}\,\frac{ \varphi^{\hbar}_{n}(S, z; x) }{ \varphi^{\hbar}_{0}(S, z; x) }
\end{equation}
with $c_{n} \defeq \sqrt{2^{|n|}n!}$ as in \eqref{eq:c_n}, and call $\{ \mathcal{P}^{\hbar}_{n}(S, z; \,\cdot\,) \}_{n \in \N_{0}^{d}}$ the {\em Hagedorn polynomials} so that
\begin{equation*}
  \varphi^{\hbar}_{n}(S, z; x) = \frac{\mathcal{P}^{\hbar}_{n}(S, z; x)}{c_{n}}\, \varphi^{\hbar}_{0}(S, z; x).
\end{equation*}

It turns out that Proposition~\ref{prop:As-as} also implies that the Hagedorn wave packets and the Hermite functions are related to each other just as in \eqref{eq:Hagedorn-Hermite_0} at {\em every} level of their ladders, not just at the ground level:
\begin{theorem}[The Hagedorn--Hermite correspondence]
  \label{thm:Hagedorn-Hermite}
  The Hagedorn wave packets
  \begin{equation*}
    \{ \varphi^{\hbar}_{n}(S,z;\,\cdot\,) \}_{n \in \N_{0}^{d}}
  \end{equation*}
  and the semiclassically scaled Hermite functions $\{ \psi^{\hbar}_{n} \}_{n \in \N_{0}^{d}}$ (see Appendix~\ref{sec:Hermite} for the definition) are related to each other as follows: For any $n \in \N_{0}^{d}$,
  \begin{equation}
    \label{eq:Hagedorn-Hermite}
    \varphi^{\hbar}_{n}(S,z;\,\cdot\,) = e^{-\frac{\rmi}{2\hbar} p\cdot q}\, \widehat{T}_{z}\, \widehat{S}\, \psi^{\hbar}_{n},
  \end{equation}
  that is, the diagrams
  \begin{equation*}
    \begin{tikzcd}[column sep=10ex, row sep=8ex, ampersand replacement=\&]
      \psi^{\hbar}_{n} \arrow[mapsto]{r}{e^{-\frac{\rmi}{2\hbar} p\cdot q}\,\widehat{T}_{z}\,\widehat{S}} \arrow[mapsto]{d}[swap]{\frac{\hat{a}_{j}}{\sqrt{n_{j}}}} \& \varphi^{\hbar}_{n}(S,z;\,\cdot\,) \arrow[mapsto]{d}{\frac{\mathscr{A}_{j}(S,z)}{\sqrt{n_{j}}}}
      \\
      \psi^{\hbar}_{n-e_{j}} \arrow[mapsto]{r}[swap]{e^{-\frac{\rmi}{2\hbar} p\cdot q}\,\widehat{T}_{z}\,\widehat{S}} \& \varphi^{\hbar}_{n-e_{j}}(S,z;\,\cdot\,)
    \end{tikzcd}
    \qquad
    \begin{tikzcd}[column sep=11ex, row sep=8ex, ampersand replacement=\&]
       \psi^{\hbar}_{n+e_{j}} \arrow[mapsto]{r}{e^{-\frac{\rmi}{2\hbar} p\cdot q}\,\widehat{T}_{z}\,\widehat{S}} \& \varphi^{\hbar}_{n+e_{j}}(S,z;\,\cdot\,)
      \\
      \psi^{\hbar}_{n} \arrow[mapsto]{r}[swap]{e^{-\frac{\rmi}{2\hbar} p\cdot q}\,\widehat{T}_{z}\,\widehat{S}} \arrow[mapsto]{u}{\frac{\hat{a}^{*}_{j}}{\sqrt{n_{j}+1}}} \&  \varphi^{\hbar}_{n}(S,z;\,\cdot\,) \arrow[mapsto]{u}[swap]{\frac{\mathscr{A}^{*}_{j}(S,z)}{\sqrt{n_{j}+1}}}
    \end{tikzcd}
  \end{equation*}
  commute for any $n \in \N_{0}^{d}$ and $j \in \{1, \dots, d\}$, where $n_{j} \ge 1$ is assumed in the left diagram.
\end{theorem}

\begin{proof}
  We know from \eqref{eq:Hagedorn-Hermite_0} that \eqref{eq:Hagedorn-Hermite} holds for $n = 0$, i.e., at the bottom of the ladders.
  Then the above diagrams follow by stacking up the diagrams \eqref{cd:As-as} from Proposition~\ref{prop:As-as}---with the operators being divided by appropriate constants $\sqrt{n_{j}}$ and $\sqrt{n_{j}+1}$ etc.---recursively (or more precisely by induction) along with the relations \eqref{eq:varphi_n-raised}, \eqref{eq:varphi_n-lowered}, and \eqref{eq:psi_n-ladders}.
\end{proof}

The above characterization of the Hagedorn wave packets can be exploited to give very simple proofs of the following fundamental facts originally due to \citet{Hagedorn1985, Ha1998}:

\begin{corollary}[{\citet[Lemma~2.1]{Hagedorn1985}; see also \citet{Ha1998}}]
  \label{cor:Hagedorn_is_cons}
  The Hagedorn wave packets $\{ \varphi^{\hbar}_{n}(S,z;\,\cdot\,) \}_{n \in \N_{0}^{d}}$ form an orthonormal basis for $L^{2}(\R^{d})$.
\end{corollary}

\begin{proof}
  The operator $\widehat{U} \defeq e^{-\frac{\rmi}{2\hbar} p\cdot q}\, \widehat{T}_{z}\, \widehat{S}$ is unitary because both the Heisenberg--Weyl operator $\widehat{T}_{z}$ and the metaplectic operator $\widehat{S}$ are unitary.
  So Theorem~\ref{thm:Hagedorn-Hermite} states that each $\varphi^{\hbar}_{n}(S,z;\,\cdot\,)$ is the result of applying the unitary operator $\widehat{U} \in \U(L^{2}(\R^{d}))$ to $\psi^{\hbar}_{n}$.
  Therefore, the Hagedorn wave packets $\{ \varphi^{\hbar}_{n}(S,z;\,\cdot\,) \}_{n \in \N_{0}^{d}}$ inherit orthonormality and completeness from the Hermite functions $\{ \psi^{\hbar}_{n} \}_{n \in \N_{0}^{d}}$.
\end{proof}

\begin{corollary}[{\citet[Lemma~2.2]{Hagedorn1985}; see also \citet[Eq.~(3.19)]{Ha1998}}]
  \label{cor:Hagedorn-Fourier}
  For any $S \in \Sp(2d,\R)$, $z \in T^{*}\R^{d}$, and multi-index $n \in \N_{0}^{d}$, 
  \begin{equation*}
    \mathscr{F}_{\hbar} \varphi^{\hbar}_{n}(S,z;\,\cdot\,) = \rmi^{d/2}\,e^{-\frac{\rmi}{\hbar} p\cdot q}\, \varphi^{\hbar}_{n}(J S, J z;\,\cdot\,),
  \end{equation*}
  or more explicitly,
  \begin{equation*}
    \mathscr{F}_{\hbar} \varphi^{\hbar}_{n}(Q, P, q, p;\,\cdot\,) = \rmi^{d/2}\,e^{-\frac{\rmi}{\hbar} p\cdot q}\, \varphi^{\hbar}_{n}(P, -Q, p, -q;\,\cdot\,).
  \end{equation*}
\end{corollary}

\begin{proof}
  Using $\mathscr{F}_{\hbar} = \rmi^{d/2} \widehat{J}$ and the symplectic covariance~\eqref{eq:T-S_covariance} of $\widehat{T}_{z}$, we have
  \begin{align*}
    \mathscr{F}_{\hbar} \varphi^{\hbar}_{n}(S, z; \,\cdot\,)
    &= e^{-\frac{\rmi}{2\hbar} p\cdot q}\, \mathscr{F}_{\hbar}\, \widehat{T}_{z}\, \widehat{S}\, \psi^{\hbar}_{n} \\
    &= e^{-\frac{\rmi}{2\hbar} p\cdot q}\, \mathscr{F}_{\hbar}\, \widehat{T}_{z}\, \mathscr{F}_{\hbar}^{*} \mathscr{F}_{\hbar}\, \widehat{S}\, \psi^{\hbar}_{n} \\
    &= \rmi^{d/2}\,e^{-\frac{\rmi}{2\hbar} p\cdot q}\, \widehat{J}\, \widehat{T}_{z}\, \widehat{J}^{*} \widehat{J}\, \widehat{S}\, \psi^{\hbar}_{n} \\
    &= \rmi^{d/2}\,e^{-\frac{\rmi}{\hbar} p\cdot q}\,e^{\frac{\rmi}{2\hbar} q\cdot p}\, \widehat{T}_{J z}\, \widehat{J}\, \widehat{S}\, \psi^{\hbar}_{n} \\
    &= \rmi^{d/2}\,e^{-\frac{\rmi}{\hbar} p\cdot q}\, \varphi^{\hbar}_{n}(J S, J z;\,\cdot\,).
  \end{align*}
  Recall that $S \mapsto J S$ and $z \mapsto J z$ correspond to $(Q, P) \mapsto (P, -Q)$ and $(q, p) \mapsto (p, -q)$, respectively.
\end{proof}

\begin{remark}
  Again, the apparent difference in the constant factors---$\rmi^{d/2}$ in our expression whereas $(-\rmi)^{|k|}$ in Lemma~2.1 of \citet{Hagedorn1985} or Eq.~(3.19) of \citet{Ha1998}---stems from different parametrizations of elements in $S \in \Sp(2d,\R)$; see Remark~\ref{eq:QP_versus_AB}.
\end{remark}

\section{Minimal Uncertainty Products for Ground State Hagedorn Wave Packet}
\label{sec:minimal_uncertainty}
The characterization of the ladder operators of Hagedorn in Section~\ref{ssec:Hagedorn_ladder_operators} is also useful in generalizing the minimal uncertainty product obtained by \citet{Ha2013} for the one-dimensional case to $d$-dimensions for any $d \in \N$.

\subsection{Symplectic Rotation of Position \& Momentum Operators}
Let us first express the operators $\hat{x} - q$ and $\hat{p} - p$ in terms of the ladder operators $( \mathscr{A}(S,z), \mathscr{A}^{*}(S,z) )$ as is done in \citet{Ha1998}.
In our setting, this is done by inverting the relation~\eqref{def:As}.
Since $\mathcal{W}$ is unitary (see \eqref{eq:mathcalW}), one obtains
\begin{equation}
  \label{eq:zhat_in_As}
  \hat{z} - z
  = \sqrt{\hbar}\, S \mathcal{W}^{*}
  \begin{bmatrix}
    \mathscr{A}(S,z) \smallskip\\
    \mathscr{A}^{*}(S,z)
  \end{bmatrix}
  = \sqrt{\frac{\hbar}{2}}
  \begin{bmatrix}
    \overline{Q} & Q \smallskip\\
    \overline{P} & P
  \end{bmatrix}
  \begin{bmatrix}
    \mathscr{A}(S,z) \smallskip\\
    \mathscr{A}^{*}(S,z)
  \end{bmatrix},
\end{equation}
or
\begin{equation*}
  \hat{x} - q = \sqrt{\frac{\hbar}{2}}\parentheses{ \overline{Q} \mathscr{A}(S,z) + Q \mathscr{A}^{*}(S,z) },
  \qquad
  \hat{p} - p = \sqrt{\frac{\hbar}{2}}\parentheses{ \overline{P} \mathscr{A}(S,z) + P \mathscr{A}^{*}(S,z) },
\end{equation*}
which are (3.28) and (3.29) in \citet{Ha1998}.

Now, consider the set of $2d$ operators $\hat{\zeta} \defeq (\hat{\xi}, \hat{\eta})$ defined as a symplectic rotation by $\mathcal{R} \in \Sp(2d,\R) \cap \Orth(2d)$ of the operators $\hat{z} - z = (\hat{x} - q, \hat{p} - p)$ in the phase space $T^{*}\R^{d} \cong \R^{2d}$, i.e.,
\begin{equation}
  \label{eq:zeta}
  \hat{\zeta} =
  \begin{bmatrix}
    \hat{\xi} \\
    \hat{\eta}
  \end{bmatrix}
  \defeq \mathcal{R}(\hat{z} - z),
\end{equation}
or equivalently, by setting
\begin{equation}
  \label{eq:mathcalR}
  \mathcal{R} \defeq
  \begin{bmatrix}
    U & V \\
    -V & U
  \end{bmatrix} \in \Sp(2d,\R) \cap \Orth(2d)
  \iff
  U^{T}V = V^{T}U
  \text{ and }
  U^{T}U + V^{T}V = I_{d},
\end{equation}
we may write 
\begin{equation*}
  \hat{\zeta} =
  \begin{bmatrix}
    \hat{\xi} \\
    \hat{\eta}
  \end{bmatrix}
  \defeq
  \begin{bmatrix}
    U & V \\
    -V & U
  \end{bmatrix}
  \begin{bmatrix}
    \hat{x} - q\\
    \hat{p} - p
  \end{bmatrix}.
\end{equation*}
Note that the intersection $\Sp(2d,\R) \cap \Orth(2d)$ may be identified with the unitary group $\U(d)$ by the map
\begin{equation*}
  \mathcal{R} = 
  \begin{bmatrix}
    U & V \\
    -V & U
  \end{bmatrix}
  \mapsto
  U + \rmi V.
\end{equation*}
It is easy to see that $\hat{\zeta} = (\hat{\xi}, \hat{\eta})$ satisfies the canonical commutation relations on $\mathscr{S}(\R^{d})$:
Let us first rewrite
\begin{equation*}
  \hat{\zeta} = (J \mathcal{R}^{T})^{T} J (\hat{z} - z) = \rho(J \mathcal{R}^{T}; \hat{z} - z).
\end{equation*}
Since $J \mathcal{R}^{T} \in \Sp(2d,\R)$ and $\hat{z} - z$ clearly satisfies the canonical commutation relations, so does $\hat{\zeta} \defeq (\hat{\xi}, \hat{\eta})$ due to Proposition~\ref{prop:Littlejohn}.

\subsection{Minimal Uncertainty Products for Ground State Hagedorn Wave Packet}
Let us introduce some shorthand notation before stating the main result of this section.
Suppose $\mathscr{B}$ is a symmetric operator with domain $D(\mathscr{B}) = \mathscr{S}(\R^{d})$ along with the property that $\mathscr{B} \psi \in \mathscr{S}(\R^{d})$ for any $\psi \in \mathscr{S}(\R^{d})$.
We introduce the following shorthand notation for the expectation value for measurements of $\mathscr{B}$ in the state $\varphi^{\hbar}_{0}(S, z;\,\cdot\,)$:
\begin{equation*}
  \exval{ \mathscr{B} }_{0} \defeq \ip{ \varphi^{\hbar}_{0}(S,z;\,\cdot\,) }{\mathscr{B}\, \varphi^{\hbar}_{0}(S,z;\,\cdot\,) }.
\end{equation*}
For example, it is easy to see that $\bigl\langle \hat{\xi}_{j} \bigr\rangle_{0} = \bigl\langle \hat{\eta}_{j} \bigr\rangle_{0} = 0$ for each $j \in \{1, 2, \dots, d\}$.
Also, let us denote by $\Delta_{0} \mathscr{B}$ the uncertainty or standard deviation associated with measurements of $\mathscr{B}$ in the state $\varphi^{\hbar}_{0}(S, z;\,\cdot\,)$, i.e.,
\begin{equation*}
  (\Delta_{0} \mathscr{B})^{2} = \bigl\langle \mathscr{B}^{2} \bigr\rangle_{0} - \bigl\langle \mathscr{B} \bigr\rangle_{0}^{2}.
\end{equation*}
So we have, for example, $(\Delta_{0} \hat{\xi}_{j})^{2} = \bigl\langle \hat{\xi}_{j}^{2} \bigr\rangle_{0}$ and $(\Delta_{0} \hat{\eta}_{j})^{2} = \bigl\langle \hat{\eta}_{j}^{2} \bigr\rangle_{0}$ for each $j \in \{1, 2, \dots, d\}$ where no summation is assumed over $j$.

\citet{Ha2013} showed in the one-dimensional case, i.e., for $d = 1$, that there exists $\mathcal{R} \in \SO(2,\R)$ such that $\Delta_{0} \hat{\xi}_{1}$ and $\Delta_{0} \hat{\eta}_{1}$ give the minimal uncertainty product, i.e.,
\begin{equation*}
  \Delta_{0}\hat{\xi}_{1}\, \Delta_{0}\hat{\eta}_{1} = \frac{\hbar}{2}.
\end{equation*}
The main result of this section generalizes this to multi-dimensions:

\begin{theorem}[Minimal uncertainty products for $\varphi^{\hbar}_{0}(S,z;\,\cdot\,)$]
  There exists a symplectic rotation matrix $\mathcal{R} \in \Sp(2d,\R) \cap \Orth(2d)$ such that, for any $j \in \{1, 2, \dots, d\}$, the uncertainty product for the pair of operators $\hat{\xi}_{j}$ and $\hat{\eta}_{j}$ (defined in \eqref{eq:zeta}) with respect to the ground state~\eqref{eq:varphi_0} of the Hagedorn wave packets is minimized, i.e.,
  \begin{equation}
    \label{eq:MinimumUncertainty}
    \Delta_{0}\hat{\xi}_{j}\, \Delta_{0}\hat{\eta}_{j} = \frac{\hbar}{2}.
  \end{equation}
\end{theorem}

\begin{proof}
  We first write the set of operators $\hat{\zeta} \defeq (\hat{\xi}, \hat{\eta})$ in terms of the ladder operators $( \mathscr{A}, \mathscr{A}^{*} )$ using \eqref{eq:zhat_in_As} and \eqref{eq:zeta}:
  \begin{equation*}
    \hat{\zeta} = 
    \begin{bmatrix}
      \hat{\xi} \\
      \hat{\eta}
    \end{bmatrix}
    = \sqrt{\hbar}\, \mathcal{R} S \mathcal{W}^{*}
    \begin{bmatrix}
      \mathscr{A} \\
      \mathscr{A}^{*}
    \end{bmatrix},
  \end{equation*}
  where $\mathcal{R} \in \Sp(2d,\R) \cap \Orth(2d)$ is arbitrary for now, and we suppressed the parameters $(S,z)$ and used $\mathscr{A}$ as a shorthand for $\mathscr{A}(S,z)$.
  Then we find
  \begin{equation*}
    \begin{bmatrix}
      \hat{\xi}\, \hat{\xi}^{T} & \hat{\xi}\, \hat{\eta}^{T} \smallskip\\
      \hat{\eta}\, \hat{\xi}^{T}  & \hat{\eta}\, \hat{\eta}^{T}
    \end{bmatrix}
    = \hat{\zeta} \hat{\zeta}^{T}
    = \hbar\, \mathcal{R} S \mathcal{W}^{*}
    \begin{bmatrix}
      \mathscr{A} \mathscr{A}^{T} & \mathscr{A} (\mathscr{A}^{*})^{T} \smallskip\\
      \mathscr{A}^{*} \mathscr{A}^{T}  & \mathscr{A}^{*} (\mathscr{A}^{*})^{T}
    \end{bmatrix}
    \overline{\mathcal{W}} S^{T} \mathcal{R}^{T},
  \end{equation*}
  where $\hat{\xi}\, \hat{\eta}^{T}$, for example, stands for the $d \times d$ matrix of operators whose $(j,k)$-component is $\hat{\xi}_{j} \hat{\eta}_{k}$ and similarly for others.
  Taking the expectation values of both sides of the above equality with respect to the ground state wave packet \eqref{eq:varphi_0}, 
  \begin{equation*}
    \begin{bmatrix}
      \texval{ \hat{\xi}\, \hat{\xi}^{T} }_{0} & \texval{ \hat{\xi}\, \hat{\eta}^{T} }_{0} \smallskip\\
      \texval{ \hat{\eta}\, \hat{\xi}^{T} }_{0}  & \texval{ \hat{\eta}\, \hat{\eta}^{T} }_{0}
    \end{bmatrix}
    = 
    \hbar\, \mathcal{R} S \mathcal{W}^{*}
    \begin{bmatrix}
      \texval{ \mathscr{A} \mathscr{A}^{T} }_{0} & \texval{ \mathscr{A} (\mathscr{A}^{*})^{T} }_{0} \smallskip\\
      \texval{ \mathscr{A}^{*} \mathscr{A}^{T} }_{0} & \texval{ \mathscr{A}^{*} (\mathscr{A}^{*})^{T} }_{0}
    \end{bmatrix}
    \overline{\mathcal{W}} S^{T} \mathcal{R}^{T}.
  \end{equation*}
  However, writing $\varphi^{\hbar}_{0} = \varphi^{\hbar}_{0}(S,z;\,\cdot\,)$ for brevity, $\mathscr{A}_{j} \varphi^{\hbar}_{0} = 0$ for any $j \in \{1, \dots, d\}$; so
  \begin{equation*}
    \mathscr{A}_{j} \mathscr{A}_{k} \varphi^{\hbar}_{0} = 0,
    \qquad
    \mathscr{A}_{j}^{*} \mathscr{A}_{k} \varphi^{\hbar}_{0} = 0
  \end{equation*}
  and
  \begin{equation*}
    \texval{ \mathscr{A}_{j}^{*} \mathscr{A}_{k}^{*} }_{0} = \ip{ \varphi^{\hbar}_{0} }{ \mathscr{A}_{j}^{*} \mathscr{A}_{k}^{*} \varphi^{\hbar}_{0} }
    = \ip{ \mathscr{A}_{j} \varphi^{\hbar}_{0} }{ \mathscr{A}_{k}^{*} \varphi^{\hbar}_{0} } = 0
  \end{equation*}
  for any $j, k \in \{1, \dots, d\}$, whereas, using the identity $[\mathscr{A}_{j}(S,z), \mathscr{A}^{*}_{k}(S,z)] = \delta_{jk}$ from \eqref{eq:As-commutators},
 \begin{equation*}
   \mathscr{A}_{j} \mathscr{A}_{k}^{*} \varphi^{\hbar}_{0} = (\delta_{jk} + \mathscr{A}_{k}^{*} \mathscr{A}_{j}) \varphi^{\hbar}_{0} = \delta_{jk}\, \varphi^{\hbar}_{0}.
  \end{equation*}
  Hence $\texval{ \mathscr{A}_{j} \mathscr{A}_{k}^{*} }_{0} = \delta_{jk}$ and as a result we have
  \begin{equation*}
    \begin{bmatrix}
      \texval{ \mathscr{A} \mathscr{A}^{T} }_{0} & \texval{ \mathscr{A} (\mathscr{A}^{*})^{T} }_{0} \smallskip\\
      \texval{ \mathscr{A}^{*} \mathscr{A}^{T} }_{0} & \texval{ \mathscr{A}^{*} (\mathscr{A}^{*})^{T} }_{0}
    \end{bmatrix}
    =
    \begin{bmatrix}
      0 & I_{d} \\
      0 & 0
    \end{bmatrix},
  \end{equation*}
  and hence, using the expression \eqref{eq:mathcalW} for $\mathcal{W}$, 
  \begin{equation*}
    \mathcal{W}^{*}
    \begin{bmatrix}
      \texval{ \mathscr{A} \mathscr{A}^{T} }_{0} & \texval{ \mathscr{A} (\mathscr{A}^{*})^{T} }_{0} \smallskip\\
      \texval{ \mathscr{A}^{*} \mathscr{A}^{T} }_{0} & \texval{ \mathscr{A}^{*} (\mathscr{A}^{*})^{T} }_{0}
    \end{bmatrix}
    \overline{\mathcal{W}}
    = \frac{1}{2} 
    \begin{bmatrix}
      I_{d} & \rmi I_{d} \\
      -\rmi I_{d} & I_{d}
    \end{bmatrix}
    = \frac{1}{2}(I_{2d} + \rmi J).
  \end{equation*}
  Therefore, 
  \begin{equation*}
    \begin{bmatrix}
      \texval{ \hat{\xi}\, \hat{\xi}^{T} }_{0} & \texval{ \hat{\xi}\, \hat{\eta}^{T} }_{0} \smallskip\\
      \texval{ \hat{\eta}\, \hat{\xi}^{T} }_{0}  & \texval{ \hat{\eta}\, \hat{\eta}^{T} }_{0}
    \end{bmatrix}
    = 
    \frac{\hbar}{2}\, \mathcal{R} S (I_{2d} + \rmi J) S^{T} \mathcal{R}^{T}
    = 
    \frac{\hbar}{2}\, (\mathcal{R} S S^{T} \mathcal{R}^{T} + \rmi J),
  \end{equation*}
  because $S, \mathcal{R} \in \Sp(2d,\R)$; that is,
  \begin{equation*}
    (\Delta_{0}\hat{\xi}_{j})^{2} = \texval{ \hat{\xi}_{j}^{2} } = \frac{\hbar}{2}\, (\mathcal{R} S S^{T} \mathcal{R}^{T})_{jj},
    \qquad
    (\Delta_{0}\hat{\eta}_{j})^{2} = \texval{ \hat{\eta}_{j}^{2} } = \frac{\hbar}{2}\, (\mathcal{R} S S^{T} \mathcal{R}^{T})_{d+j,d+j}
  \end{equation*}
  for each $j \in \{1, 2, \dots, d\}$ (no summation is assumed over $j$).

  Now, since the matrix $S S^{T}$ is positive-definite and symplectic, we may choose $\mathcal{R} \in \Sp(2d,\R) \cap \Orth(2d)$ such that
  \begin{equation*}
    \mathcal{R} S S^{T} \mathcal{R}^{T} = \diag(\lambda_{1}, \dots, \lambda_{d}, 1/\lambda_{1}, \dots, 1/\lambda_{d}),
  \end{equation*}
  where $\lambda_{j} > 0$ for each $j \in \{1, 2, \dots, d\}$ (see, e.g., \citet[Proposition~32 on p.~26]{Go2011}).
  As a result we obtain, for each $j \in \{1, 2, \dots, d\}$, 
  \begin{equation*}
    (\Delta_{0}\hat{\xi}_{j})^{2} = \frac{\hbar}{2}\,\lambda_{j},
    \qquad
    (\Delta_{0}\hat{\eta}_{j})^{2} = \frac{\hbar}{2}\,\lambda_{j}^{-1},
  \end{equation*}
  which implies the minimum uncertainty relation~\eqref{eq:MinimumUncertainty}.
\end{proof}

\begin{example}[The one-dimensional case; {\citet{Ha2013}}]
  Consider the one-dimensional case, i.e., $d = 1$.
  The matrix $S$ in \eqref{eq:S_in_QP} is $2 \times 2$ with $Q, P \in \C$, and $\mathcal{R}$ in \eqref{eq:mathcalR} is in $\Sp(2,\R) \cap \Orth(2,\R) = \SO(2,\R)$ and thus can be written as
  \begin{equation*}
    \mathcal{R} =
    \begin{bmatrix}
      \cos\theta & \sin\theta \\
      -\sin\theta & \cos\theta
    \end{bmatrix}.
  \end{equation*}
  However, from the last step of the above proof, we know that the minimal uncertainty relation~\eqref{eq:MinimumUncertainty} is realized if the row vectors of $\mathcal{R}$ are the normalized eigenvectors of $S S^{T}$.
  Tedious but straightforward calculations of these eigenvectors yield
  \begin{equation*}
    \tan(2\theta) = \frac{2\Re(P\overline{Q})}{|Q|^{2} - |P|^{2}} = \frac{2\Im(B\overline{A})}{|B|^{2} - |A|^{2}},
  \end{equation*}
  where $A = Q$ and $B = -\rmi P$ is the notation of \citet{Ha1980,Ha1998,Ha2013}, and $|Q| \defeq (Q Q^{*})^{1/2}$.
  This is precisely Theorem~5.2 of \citet{Ha2013}.
\end{example}

\section{Generating Function for the Hagedorn Wave Packets}
\label{sec:generating_functions}
In Theorem~\ref{thm:Hagedorn-Hermite}, we established a link between the Hagedorn wave packets and Hermite functions using a unitary operator essentially consisting of the Heisenberg--Weyl and metaplectic operators.
This simple link suggests that those properties satisfied by the Hermite functions may also be adapted into the corresponding ones for the Hagedorn wave packets by means of the unitary operator.
One such example is the simple proof of Corollary~\ref{cor:Hagedorn_is_cons} that the Hagedorn wave packets form an orthonormal basis for $L^{2}(\R^{d})$.

As another example, this section takes the generating functions for the Hermite functions and polynomials and shows how they can be transformed into the generating functions for the Hagedorn wave packets and polynomials.
Such generating functions are obtained in \citet{DiKeTr2017} and \citet{Ha2015}.
We present an alternative derivation of them based on Theorem~\ref{thm:Hagedorn-Hermite} using the Heisenberg--Weyl and metaplectic operators.
Our derivation reveals how the generating functions of Hagedorn and Hermite are related to each other, and shows that the former follows from the latter.

\subsection{Generating Functions for the Hermite Functions and Polynomials}
Let us first briefly review the generating functions for the Hermite functions and polynomials.
See Appendix~\ref{sec:Hermite} for a more detailed account.
The semiclassically scaled Hermite functions $\{ \psi^{\hbar}_{n} \}_{n \in \N_{0}^{d}}$ are given as
\begin{equation*}
  \psi^{\hbar}_{n}(x)
  = \frac{p^{\hbar}_{n}(x)}{\sqrt{2^{|n|}n!}}\,\psi^{\hbar}_{0}(x),
\end{equation*}
where $\psi^{\hbar}_{0}$ is the ground state~\eqref{eq:psi_0} and $\{ p^{\hbar}_{n} \}_{n \in \N_{0}^{d}}$ are the {\em semiclassically scaled Hermite polynomials}; see \eqref{eq:p^hbar_n}.
It is well known that 
\begin{align}
  \label{eq:Gamma-Hermite}
  \Gamma^{\hbar}(w,x)
  \defeq \frac{1}{(\pi\hbar)^{d/4}}\,\exp\parentheses{ -\frac{x^{2}}{2\hbar} + \frac{2}{\sqrt{\hbar}}\,w^{T}x - w^{2} }
  = \sum_{n \in \N_{0}^{d}} \psi^{\hbar}_{n}(x)\,\frac{c_{n}}{n!}\,w^{n}
\end{align}
and
\begin{equation}
  \label{eq:gamma-Hermite}
  \gamma^{\hbar}(w,x)
  \defeq \frac{\gamma^{\hbar}(w, x)}{\psi^{\hbar}_{0}(x)}
  = \exp\parentheses{ \frac{2}{\sqrt{\hbar}}\,w^{T}x - w^{2} }
  = \sum_{n \in \N_{0}^{d}} p^{\hbar}_{n}(x)\,\frac{w^{n}}{n!},
\end{equation}
where $w \in \C^{n}$, and the coefficients $\{ c_{n} \}_{n \in \N_{0}^{d}}$ are defined in \eqref{eq:c_n}; hence we may call  $\Gamma^{\hbar}(w,x)$ and $\gamma^{\hbar}(w,x)$ the {\em generating functions} for the Hermite functions and Hermite polynomials, respectively.

\subsection{Transformation of Generating Function $\Gamma^{\hbar}$}
Now we would like to derive the generating functions for the Hagedorn wave packets and Hagedorn polynomials using the same techniques and tools as in Section~\ref{sec:Hagedorn_wave_packets}.

Based on what we have in Theorem~\ref{thm:Hagedorn-Hermite}, it is natural to conjecture that $e^{-\frac{\rmi}{2\hbar} p\cdot q}\, \widehat{T}_{z}\, \widehat{S}\, \Gamma^{\hbar}(w,\,\cdot\,)$ would give the generating function for the Hagedorn wave packets.
In fact, applying the operator $e^{-\frac{\rmi}{2\hbar} p\cdot q}\, \widehat{T}_{z}\, \widehat{S}$ to both sides of \eqref{eq:Gamma-Hermite}, we have, using \eqref{eq:Hagedorn-Hermite} in Theorem~\ref{thm:Hagedorn-Hermite},
\begin{equation*}
  e^{-\frac{\rmi}{2\hbar} p\cdot q}\, \widehat{T}_{z}\, \widehat{S}\, \Gamma^{\hbar}(w,\,\cdot\,)
  = \sum_{n \in \N_{0}^{d}} \varphi^{\hbar}_{n}(S, z;\,\cdot\,)\,\frac{c_{n}}{n!}\,w^{n},
\end{equation*}
because the operator $e^{-\frac{\rmi}{2\hbar} p\cdot q}\, \widehat{T}_{z}\, \widehat{S}$ is unitary and thus applies to the series on the right-hand side term by term.
Therefore, the problem boils down to finding an explicit expression of the function on the left-hand side of the above equation; the resulting function gives the generating function for the Hagedorn wave packets.

Finding an expression for the generating function---particularly the calculation of $\widehat{S}\,\Gamma^{\hbar}(w,\,\cdot\,)$---is a little tricky, because the metaplectic operators $\widehat{S} \in \Mp(2d,\R)$ are not always given in simple explicit forms as we mentioned in Section~\ref{ssec:Mp} (particularly Remark~\ref{rem:abstractness_of_Mp}).
Therefore, we first would like to find $\widehat{S}\,\Gamma^{\hbar}(w,\,\cdot\,)$ for the special case where $S \in \FSp(2d,\R)$, i.e., $S$ is a free symplectic matrix (see the definition~\eqref{eq:FSp} in Section~\ref{ssec:Mp}), because in this case $\widehat{S}$ is given explicitly as a quadratic Fourier transform~\eqref{eq:quad_Fourier_transform2}.
Let $S \in \FSp(2d,\R)$ and write, as in \eqref{eq:S_in_QP}, 
\begin{equation*}
  S =
  \begin{bmatrix}
    \Re Q & \Im Q \smallskip\\
    \Re P & \Im P
  \end{bmatrix},
\end{equation*}
where $\Im Q \neq 0$ is assumed by definition.
Then, using \eqref{eq:quad_Fourier_transform2} and evaluating the resulting Gaussian integral, we obtain
\begin{equation}
  \label{eq:Gamma-FSp}
  \widehat{S}\,\Gamma^{\hbar}(w, x)
  \defeq [\widehat{S}\,\Gamma^{\hbar}(w, \,\cdot\,)](x)
  = \frac{\mu(S,\rmi I_{d})}{(\pi\hbar)^{d/4}} \exp\parentheses{ \frac{\rmi}{2\hbar}x^{T} P Q^{-1} x + \frac{2}{\sqrt{\hbar}}\,w^{T} Q^{-1} x - w^{T} Q^{-1} \overline{Q} w },
\end{equation}
where $\mu\colon \Sp(2d,\R) \times \Sigma_{d} \to \C$ is defined as
\begin{equation}
  \label{eq:mu}
  \mu\parentheses{
    \begin{bmatrix}
      A & B \\
      C & D
    \end{bmatrix},
    \mathcal{Z}
    }
    \defeq [\det(A + B \mathcal{Z})]^{-1/2}
\end{equation}
so that $\mu(S,\rmi I_{d}) = (\det Q)^{-1/2}$.
Again, the sign of $\mu$ changes depending on the branch chosen for the square root.
The definition of the factor $\mu$ is a variant of that of \citet[Eq.~(4.61) on p.~201]{Fo1989} but $\mu$ retains the same key property: It is straightforward to show that, for any $S_{j} \in \Sp(2d,\R)$ with $j = 1, 2$ and $\mathcal{Z} \in \Sigma_{d}$, we have
\begin{equation*}
  \mu(S_{1} S_{2}, \mathcal{Z}) = \mu(S_{1}, \Psi_{S_{2}}(\mathcal{Z}))\, \mu(S_{2}, \mathcal{Z}),
\end{equation*}
where $\Psi$ is the action defined in \eqref{eq:Psi}, and we need to interpret the square roots with proper branch cuts.

This motivates us to {\em define} for {\em any} $S \in \Sp(2d,\R)$,
\begin{equation}
  \label{eq:Gamma-generalized}
  \Gamma^{\hbar}(S; w, x)
  \defeq \frac{\mu(S,\rmi I_{d})}{(\pi\hbar)^{d/4}} \exp\parentheses{ \frac{\rmi}{2\hbar}x^{T} P Q^{-1} x + \frac{2}{\sqrt{\hbar}}\,w^{T} Q^{-1} x - w^{T} Q^{-1} \overline{Q} w },
\end{equation}
where we note that both $P$ and $Q$ are invertible if $S = \begin{tbmatrix}
  \Re Q & \Im Q \smallskip\\
  \Re P & \Im P
\end{tbmatrix} \in \Sp(2d,\R)$.
The above definition generalizes the generating function $\Gamma^{\hbar}(w, x)$ for the Hermite functions because $\Gamma^{\hbar}(I_{2d}; w, x) = \Gamma^{\hbar}(w, x)$.
Clearly, if $S \in \FSp(2d,\R)$ then $\Gamma^{\hbar}(S; w, x) = \widehat{S}\,\Gamma^{\hbar}(w, x)$ by definition; but then we would like to show that it is the case for {\em any} $S \in \Sp(2d,\R)$ so that \eqref{eq:Gamma-FSp} holds for {\em any} $S \in \Sp(2d,\R)$, i.e., $\widehat{S}\,\Gamma^{\hbar}(w, x) = \Gamma^{\hbar}(S; w, x)$ for {\em any} $S \in \Sp(2d,\R)$.
To that end, we first prove the following lemma; it is slightly more general than what we need, but may be thought of as the symplectic covariance of the generating function~\eqref{eq:Gamma-generalized}:

\begin{lemma}[Symplectic covariance of generating function $\Gamma^{\hbar}$]
  \label{lem:Mp-gen_fcn}
  Let $S_{0} \in \Sp(2d,\R)$ and $\widehat{S} \in \Mp(2d,\R)$ with $S \defeq \pi_{\Mp}(\widehat{S}) \in \Sp(2d,\R)$.
  Then, 
  \begin{equation*}
    \widehat{S}\,\Gamma^{\hbar}(S_{0}; w, x) = \Gamma^{\hbar}(S S_{0}; w, x).
  \end{equation*}
\end{lemma}

\begin{remark}
  The above expression $\widehat{S}_{0} \Gamma^{\hbar}(S; w, x)$ is a shorthand for $[\widehat{S}_{0} \Gamma^{\hbar}(S; w, \,\cdot\,)](x)$.
  We will use similar shorthands below for notational simplicity.
\end{remark}

\begin{proof}
  Recall from Section~\ref{ssec:Mp} that the metaplectic group $\Mp(2d,\R)$ is generated by $\widehat{J}$, $\widehat{M}_{L}$, and $\widehat{V}_{R}$ with $L \in \GL(d,\R)$ and $R \in \Sym(d,\R)$.
  So it suffices to prove the above assertion for those cases where $\widehat{S}$ is $\widehat{J}$, $\widehat{M}_{L}$, and $\widehat{V}_{R}$ for any $L \in \GL(d,\R)$ and $R \in \Sym(d,\R)$.

  First set $\widehat{S} = \widehat{J}$.
  We would like to show that $\widehat{J} \Gamma^{\hbar}(S_{0}; w, x) = \Gamma^{\hbar}(J S_{0}; w, x)$.
  Let us first evaluate $\widehat{J} \Gamma^{\hbar}(S_{0}; w, x)$.
  We have, using \eqref{eq:Jhat},
  \begin{equation*}
    \widehat{J} \Gamma^{\hbar}(S_{0}; w, x) = \frac{1}{ (2\pi\hbar\,\rmi)^{d/2} } \int_{\R^{d}} e^{-\frac{\rmi}{\hbar}x \cdot \tilde{x}}\, \Gamma^{\hbar}(S_{0};w,\tilde{x})\,d\tilde{x},
  \end{equation*}
  but then, using the expression \eqref{eq:Gamma-generalized}, the integrand becomes 
  \begin{multline*}
    e^{-\frac{\rmi}{\hbar}x \cdot \tilde{x}}\, \Gamma^{\hbar}(S_{0};w,\tilde{x}) \\
    = \frac{\parentheses{ \det Q_{0} }^{-1/2}}{(\pi\hbar)^{d/4}} \exp\brackets{ -\pi \tilde{x}^{T} \parentheses{ \frac{1}{2\pi\hbar\,\rmi}\,P_{0} Q_{0}^{-1} } \tilde{x}
      - 2\pi\rmi \parentheses{
      \frac{x}{2\pi\hbar} + \frac{\rmi}{\pi\sqrt{\hbar}}\,(Q_{0}^{-1})^{T}w
      } \cdot \tilde{x}
      - w^{T} Q_{0}^{-1} \overline{Q}_{0} w }.
  \end{multline*}
  Carrying out the integral (see \citet[Theorem~1 on p.~256]{Fo1989} for a useful formula for such Gaussian integrals) gives
  \begin{align*}
    \widehat{J} \Gamma^{\hbar}&(S_{0}; w, x) \\
                              &= \frac{\exp(- w^{T} Q_{0}^{-1} \overline{Q}_{0} w)}{(2\pi\hbar\,\rmi)^{d/2} (\pi\hbar)^{d/4}}
    \parentheses{ \det Q_{0} }^{-1/2} \brackets{ \det\parentheses{ \frac{1}{2\pi\hbar\,\rmi}\,P_{0}Q_{0}^{-1} } }^{-1/2} \\
                              &\quad\times
                                \exp\brackets{
    -\pi \parentheses{ \frac{x}{2\pi\hbar} + \frac{\rmi}{\pi\sqrt{\hbar}}\,(Q_{0}^{-1})^{T}w }^{T}
                                \parentheses{ \frac{1}{2\pi\hbar\,\rmi}\,P_{0}Q_{0}^{-1} }^{-1}
                                \parentheses{ \frac{x}{2\pi\hbar} + \frac{\rmi}{\pi\sqrt{\hbar}}\,(Q_{0}^{-1})^{T}w }
                                } \\
                              &= \frac{(\det P_{0})^{-1/2}}{(\pi\hbar)^{d/4}}
                                \exp\braces{ \frac{\rmi}{2\hbar}x^{T} (-Q_{0} P_{0}^{-1}) x + \frac{2}{\sqrt{\hbar}}\,w^{T} P_{0}^{-1} x + w^{T} \brackets{2\rmi P_{0}^{-1} (Q_{0}^{-1})^{T} - Q_{0}^{-1} \overline{Q}_{0}} w }.
  \end{align*}
  Now for the last term, recall from \eqref{def:Sp2d-block-Lubich} that $S_{0} = \begin{tbmatrix}
  \Re Q_{0} & \Im Q_{0} \smallskip\\
  \Re P_{0} & \Im P_{0}
\end{tbmatrix} \in \Sp(2d,\R)$ implies that $Q_{0}^{*}P_{0} - P_{0}^{*}Q_{0} = 2\rmi I_{d}$.
  Taking the transpose of it and multiplying both sides from the left by $P_{0}^{-1}(Q_{0}^{-1})^{T}$ gives
  \begin{equation*}
    P_{0}^{-1} (Q_{0}^{-1})^{T}P_{0}^{T} \overline{Q}_{0} - P_{0}^{-1}\overline{P}_{0} = 2\rmi P_{0}^{-1}(Q_{0}^{-1})^{T}.
  \end{equation*}
  But then $P_{0}Q_{0}^{-1} \in \Sigma_{d}$ as we have seen in \eqref{eq:PQinv}, and so $(Q_{0}^{-1})^{T}P_{0}^{T} = (P_{0}Q_{0}^{-1})^{T} = P_{0}Q_{0}^{-1}$; hence,
  \begin{equation*}
   Q_{0}^{-1} \overline{Q}_{0} - P_{0}^{-1}\overline{P}_{0} = 2\rmi P_{0}^{-1}(Q_{0}^{-1})^{T},
  \end{equation*}
  and so $2\rmi P_{0}^{-1}(Q_{0}^{-1})^{T} - Q_{0}^{-1} \overline{Q}_{0} = P_{0}^{-1}\overline{P}_{0}$.
  Therefore, we have
  \begin{equation*}
    \widehat{J} \Gamma^{\hbar}(S_{0}; w, x) 
    = \frac{(\det P_{0})^{-1/2}}{(\pi\hbar)^{d/4}}
      \exp\brackets{ \frac{\rmi}{2\hbar}x^{T} (-Q_{0} P_{0}^{-1}) x + \frac{2}{\sqrt{\hbar}}\,w^{T} P_{0}^{-1} x - w^{T} P_{0}^{-1} \overline{P}_{0} w }.
  \end{equation*}
  How about $\Gamma^{\hbar}(J S_{0}; w, x)$?
  We have
  \begin{equation*}
    J S_{0} =
    \begin{bmatrix}
      0 & I_{d} \smallskip\\
      -I_{d} & 0
    \end{bmatrix}
    \begin{bmatrix}
      \Re Q_{0} & \Im Q_{0} \smallskip\\
      \Re P_{0} & \Im P_{0}
    \end{bmatrix}
    =
    \begin{bmatrix}
      \Re P_{0} & \Im P_{0} \smallskip\\
      -\Re Q_{0} & -\Im Q_{0}
    \end{bmatrix},
  \end{equation*}
  and thus $\mu(J S_{0}, \rmi I_{d}) = (\det P_{0})^{-1/2}$; so it is easy to see from the definition~\eqref{eq:Gamma-generalized} that $\Gamma^{\hbar}(J S_{0}; w, x)$ takes the same form as the above expression for $\widehat{J} \Gamma^{\hbar}(S_{0}; w, x)$; hence $\widehat{J} \Gamma^{\hbar}(S_{0}; w, x) = \Gamma^{\hbar}(J S_{0}; w, x)$.
  
  Next set $\widehat{S} = \widehat{V}_{R}$ with any $R \in \Sym(d,\R)$.
  It is easy to see that, using \eqref{eq:Vhat}
  \begin{equation*}
    \widehat{V}_{R} \Gamma^{\hbar}(S_{0}; w, x)
    = \frac{(\det Q_{0})^{-1/2}}{(\pi\hbar)^{d/4}} \exp\brackets{
      \frac{\rmi}{2\hbar}x^{T} (R + P_{0} Q_{0}^{-1}) x + \frac{2}{\sqrt{\hbar}}\,w^{T} Q_{0}^{-1} x - w^{T} Q_{0}^{-1} \overline{Q}_{0} w
    }.
  \end{equation*}
  On the other hand,
  \begin{equation*}
    V_{R} S_{0} = 
    \begin{bmatrix}
      I_{d} & 0 \smallskip\\
      R & I_{d}
    \end{bmatrix}
    \begin{bmatrix}
      \Re Q_{0} & \Im Q_{0} \smallskip\\
      \Re P_{0} & \Im P_{0}
    \end{bmatrix} \\
    =
    \begin{bmatrix}
      \Re Q_{0} & \Im Q_{0} \smallskip\\
      R \Re Q_{0} + \Re P_{0} & R \Im Q_{0} + \Im P_{0}
    \end{bmatrix},
  \end{equation*}
  and thus $\mu(V_{R} S_{0}, \rmi I_{d}) = (\det Q_{0})^{-1/2}$; hence, using \eqref{eq:Gamma-generalized}, $\Gamma^{\hbar}(V_{R} S_{0}; w, x)$ yields the same expression as the one above for $\widehat{V}_{R} \Gamma^{\hbar}(S_{0}; w, x)$; hence $\widehat{V}_{R} \Gamma^{\hbar}(S_{0}; w, x) = \Gamma^{\hbar}(V_{R} S_{0}; w, x)$.
  
  Finally, with $\widehat{S} = \widehat{M}_{L}$ for any $L \in \GL(d,\R)$, we have, using \eqref{eq:Mhat2}, 
  \begin{equation*}
    \widehat{M}_{L} \Gamma^{\hbar}(S_{0}; w, x)
    = \frac{[\det (L^{-1} Q_{0})]^{-1/2}}{(\pi\hbar)^{d/4}} \exp\parentheses{
      \frac{\rmi}{2\hbar}x^{T} L^{T} P_{0} Q_{0}^{-1} L x + \frac{2}{\sqrt{\hbar}}\,w^{T} Q_{0}^{-1} L x - w^{T} Q_{0}^{-1} \overline{Q}_{0} w
    },
  \end{equation*}
   whereas 
  \begin{equation*}
    M_{L} S_{0} =
    \begin{bmatrix}
      L^{-1} & 0 \smallskip\\
      0 & L^{T}
    \end{bmatrix}
    \begin{bmatrix}
      \Re Q_{0} & \Im Q_{0} \smallskip\\
      \Re P_{0} & \Im P_{0}
    \end{bmatrix} \\
    =
    \begin{bmatrix}
      L^{-1}\Re Q_{0} & L^{-1}\Im Q_{0} \smallskip\\
      L^{T}\Re P_{0} & L^{T}\Im P_{0}
    \end{bmatrix},
  \end{equation*}
  and thus $\mu(M_{L} S_{0}, \rmi I_{d}) = [\det(L^{-1} Q_{0})]^{-1/2}$; so it is easy to see, using \eqref{eq:Gamma-generalized}, that $\widehat{M}_{L} \Gamma^{\hbar}(S_{0}; w, x) = \Gamma^{\hbar}(M_{L} S_{0}; w, x)$.
\end{proof}

\subsection{The Generating Function for the Hagedorn Wave Packets}
It is now easy to prove the main result of this section:
\begin{theorem}
  Let $\Gamma^{\hbar}(w,\,\cdot\,)$ be the generating function~\eqref{eq:Gamma-Hermite} for the Hermite functions, and define the function $\Gamma^{\hbar}(S, z; w,\,\cdot\,) \in \mathscr{S}(\R^{d})$ with $S =
  \begin{tbmatrix}
    \Re Q & \Im Q \smallskip\\
    \Re P & \Im P
  \end{tbmatrix} \in \Sp(2d,\R)$, $z \in T^{*}\R^{d}$, and $w \in \C^{d}$ as
  \begin{equation}
    \label{eq:Gamma-Hagedorn}
    \Gamma^{\hbar}(S, z; w,\,\cdot\,) \defeq e^{-\frac{\rmi}{2\hbar} p\cdot q}\, \widehat{T}_{z}\, \widehat{S}\, \Gamma^{\hbar}(w,\,\cdot\,).
  \end{equation}
  Then it takes the form
  \begin{equation}
    \label{eq:Gamma-Hagedorn2}
    \Gamma^{\hbar}(S, z; w, x) = \varphi^{\hbar}_{0}(S, z; x) \exp\parentheses{ \frac{2}{\sqrt{\hbar}}\,w^{T} Q^{-1} (x - q) - w^{T} Q^{-1} \overline{Q} w },
  \end{equation}
  and is the generating function for the Hagedorn wave packets $\{ \varphi^{\hbar}_{n}(S,z;\,\cdot\,) \}_{n \in \N_{0}^{d}}$, i.e.,
  \begin{equation*}
    \Gamma^{\hbar}(S, z; w, x) = \sum_{n \in \N_{0}^{d}} \varphi^{\hbar}_{n}(S, z; x)\,\frac{w^{n}}{n!}.
  \end{equation*}
  Equivalently,
  \begin{equation}
    \label{eq:gamma-Hagedorn}
    \gamma^{\hbar}(S, z; w, x) \defeq \frac{\Gamma^{\hbar}(S, z; w, x)}{\varphi^{\hbar}_{0}(S, z; x)}
    = \exp\parentheses{ \frac{2}{\sqrt{\hbar}}\,w^{T} Q^{-1} (x - q) - w^{T} Q^{-1} \overline{Q} w }
  \end{equation}
  is the generating function for the Hagedorn polynomials $\{ \mathcal{P}^{\hbar}_{n}(S, z; \,\cdot\,) \}_{n \in \N_{0}^{d}}$, i.e.,
  \begin{equation*}
    \gamma^{\hbar}(S, z; w, x) = \sum_{n \in \N_{0}^{d}} \mathcal{P}^{\hbar}_{n}(S, z; \,\cdot\,)\,\frac{w^{n}}{n!}.
  \end{equation*}
\end{theorem}

\begin{remark}
  Again, strictly speaking, there are two expressions for \eqref{eq:Gamma-Hagedorn2} that differ by the sign, depending on how one takes the branch cut in $\mu(S,\rmi I_{d}) = (\det Q)^{-1/2}$ of $\varphi^{\hbar}_{0}(S,z;\,\cdot\,)$; see \eqref{eq:varphi_0}.
\end{remark}

\begin{proof}
  First it is easy to see that setting $S_{0} = I_{2d}$ in Lemma~\ref{lem:Mp-gen_fcn} implies that $\widehat{S}\,\Gamma^{\hbar}(w, x) = \Gamma^{\hbar}(S; w, x)$ holds for any $S \in \Sp(2d,\R)$, and therefore, using \eqref{eq:Gamma-generalized}, we have
  \begin{equation*}
    \widehat{S}\,\Gamma^{\hbar}(w, x) = \Gamma^{\hbar}(S; w, x) = \frac{\mu(S,\rmi I_{d})}{(\pi\hbar)^{d/4}} \exp\parentheses{ \frac{\rmi}{2\hbar}x^{T} P Q^{-1} x + \frac{2}{\sqrt{\hbar}}\,w^{T} Q^{-1} x - w^{T} Q^{-1} \overline{Q} w }
  \end{equation*}
  for any $S \in \Sp(2d,\R)$.
  Then the expression \eqref{eq:Gamma-Hagedorn2} follows easily from the definition~\eqref{eq:Gamma-Hagedorn}:
  \begin{align*}
    \Gamma^{\hbar}(S, z; w,x)
    &\defeq e^{-\frac{\rmi}{2\hbar} p\cdot q}\, \widehat{T}_{z}\, \widehat{S}\, \Gamma^{\hbar}(w,x) \\
    &= e^{-\frac{\rmi}{2\hbar} p\cdot q}\, \widehat{T}_{z}\, \Gamma^{\hbar}(S; w,x) \\
    &= \frac{\mu(S,\rmi I_{d})}{(\pi\hbar)^{d/4}}\, e^{-\frac{\rmi}{2\hbar} p\cdot q}\, \widehat{T}_{z} \brackets{
      \exp\parentheses{ \frac{\rmi}{2\hbar}x^{T} P Q^{-1} x }
      \exp\parentheses{ \frac{2}{\sqrt{\hbar}}\,w^{T} Q^{-1} x - w^{T} Q^{-1} \overline{Q} w }
      } \\
    &= \varphi^{\hbar}_{0}(S, z; x)
      \exp\parentheses{ \frac{2}{\sqrt{\hbar}}\,w^{T} Q^{-1} (x - q) - w^{T} Q^{-1} \overline{Q} w },
  \end{align*}
  where we used the following identity in the last equality:
  \begin{equation*}
    \varphi^{\hbar}_{0}(S, z; x) = \frac{\mu(S,\rmi I_{d})}{(\pi\hbar)^{d/4}}\, e^{-\frac{\rmi}{2\hbar} p\cdot q}\,
    \widehat{T}_{z} \brackets{
      \exp\parentheses{ \frac{\rmi}{2\hbar}x^{T} P Q^{-1} x }
    },
  \end{equation*}
  which is easy to verify using \eqref{eq:varphi_0} and \eqref{eq:mu} along with \eqref{eq:That}.

  Now recall the generating function~\eqref{eq:Gamma-Hermite} of the Hermite polynomials, i.e.,
  \begin{equation*}
    \Gamma^{\hbar}(w, x)
    = \sum_{n \in \N_{0}^{d}} \psi^{\hbar}_{n}(x)\,\frac{c_{n}}{n!}\,w^{n},
  \end{equation*}
  and let us apply the operator $e^{-\frac{\rmi}{2\hbar} p\cdot q}\, \widehat{T}_{z}\, \widehat{S}$ to both sides.
  As mentioned in the beginning of the section, this operator is unitary and thus applies to the series on the right-hand side term by term, i.e.,
  \begin{equation*}
    \Gamma^{\hbar}(S, z;w,x)
    = \sum_{n \in \N_{0}^{d}} e^{-\frac{\rmi}{2\hbar} p\cdot q}\, \widehat{T}_{z}\, \widehat{S}\,\psi^{\hbar}_{n}(x)\,\frac{c_{n}}{n!}\,w^{n} \\
    = \sum_{n \in \N_{0}^{d}} \varphi^{\hbar}_{n}(S, z; x)\,\frac{c_{n}}{n!}\,w^{n},
  \end{equation*}
  where we used \eqref{eq:Hagedorn-Hermite} from Theorem~\ref{thm:Hagedorn-Hermite}.
  Dividing both sides by $\varphi^{\hbar}_{0}(S, z; x)$, we have,
  \begin{equation*}
    \gamma^{\hbar}(S, z; w, x)
    = \sum_{n \in \N_{0}^{d}} c_{n}\,\frac{ \varphi^{\hbar}_{n}(S, z; x) }{ \varphi^{\hbar}_{0}(S, z; x) } \frac{w^{n}}{n!} \\
    = \sum_{n \in \N_{0}^{d}} \mathcal{P}^{\hbar}_{n}(S, z; x)\,\frac{w^{n}}{n!},
  \end{equation*}
  where we used the definition~\eqref{eq:P_n} of the Hagedorn polynomials $\{ \mathcal{P}^{\hbar}_{n}(S, z; \,\cdot\,) \}_{n \in \N_{0}^{d}}$.
\end{proof}

We may now exploit the generating function~\eqref{eq:Gamma-Hagedorn} to find the relationship between the Hagedorn polynomials $\{ \mathcal{P}^{\hbar}_{n}(S, z; \,\cdot\,) \}_{n \in \N_{0}^{d}}$ and the Hermite polynomials $\{ p^{\hbar}_{n}(\,\cdot\,) \}_{n \in \N_{0}^{d}}$:
\begin{corollary}
  \label{cor:Hagedorn-Hermite_polynomials}
  For each $n \in \N_{0}^{d}$, the Hagedorn polynomial $\mathcal{P}^{\hbar}_{n}(S, z; x)$ is written in terms of the Hermite polynomials of the $|n|$-th excited states as follows:
  \begin{equation*}
    \mathcal{P}^{\hbar}_{n}(S, z; x) = \sum_{\substack{k \in \N_{0}^{d}\\|k|=|n|}} \frac{n!}{k!}\, f^{k}_{n}(Q)\, p^{\hbar}_{k}\parentheses{ |Q|^{-1}(x - q) },
  \end{equation*}
  where $|Q| \defeq (Q Q^{*})^{1/2}$ and the coefficients $\{ f^{k}_{n}(Q) \}_{k, n \in \N_{0}^{d}}$ are defined such that, for any $k \in \N_{0}^{d}$ and any $n \in \N_{0}^{d}$ with $|n| = |k|$,     
  \begin{equation}
    \label{eq:f}
    (|Q|^{-1}\overline{Q}w)^{k} = \sum_{n \in \N_{0}^{d}} f^{k}_{n}(Q)\,w^{n},
  \end{equation}
  and $f^{k}_{n}(Q) = 0$ if $|k| \neq |n|$.
\end{corollary}

\begin{proof}
One sees from \eqref{eq:gamma-Hermite} and \eqref{eq:gamma-Hagedorn} that these generating functions are related to each other as follows:
\begin{equation*}
  \gamma^{\hbar}(S, z; w, x) = \gamma^{\hbar}\parentheses{ |Q|^{-1}\overline{Q}w , |Q|^{-1}(x - q) },
\end{equation*}
where we defined $|Q| \defeq (Q Q^{*})^{1/2}$ as in \citet{Ha2015}; note that $S \in \Sp(2d,\R)$ implies that $Q Q^{*}$ is positive-definite.
As a result, we find
\begin{equation*}
  \sum_{n \in \N_{0}^{d}} \mathcal{P}^{\hbar}_{n}(S, z; x)\,\frac{w^{n}}{n!} 
  = \sum_{k \in \N_{0}^{d}} p^{\hbar}_{k}\parentheses{ |Q|^{-1}(x - q) } \frac{(|Q|^{-1}\overline{Q}w)^{k}}{k!},
\end{equation*}
but then the above definition of $f^{k}_{n}(Q)$ yields
\begin{equation}
  \label{eq:double_sum}
  \sum_{n \in \N_{0}^{d}} \mathcal{P}^{\hbar}_{n}(S, z; x)\,\frac{w^{n}}{n!} 
  = \sum_{k \in \N_{0}^{d}} \sum_{n \in \N_{0}^{d}} f^{k}_{n}(Q)\, \frac{ p^{\hbar}_{k}\parentheses{ |Q|^{-1}(x - q) } }{k!}\, w^{n}.
\end{equation}

Let us now show that the series on right-hand side converges absolutely in a neighborhood of $w = 0$.
First, setting $w = (1, \dots, 1)$ in the definition~\eqref{eq:f} of $f^{k}_{n}(Q)$, one obtains the estimate
\begin{equation*}
  \sum_{n \in \N_{0}^{d}} \bigl| f^{k}_{n}(Q) \bigr| \le d^{|k|} \norm{ |Q|^{-1}\overline{Q} }_{\infty}^{|k|}
  = \parentheses{ d\, \norm{ |Q|^{-1}\overline{Q} }_{\infty} }^{|k|}.
\end{equation*}
Also, using the estimate~\eqref{eq:p_n-estimate} for the Hermite polynomials, we have, for any $r > 0$,
\begin{equation*}
  \bigl| p^{\hbar}_{k}\parentheses{ |Q|^{-1}(x - q) } \bigr|
  \le \frac{k!}{r^{|k|}} \exp\parentheses{ d\,r^{2} + \frac{2r}{\sqrt{\hbar}} \norm{ |Q|^{-1}(x - q) }_{1} }
  = K(d, r, Q, \hbar, x - q)\, \frac{k!}{r^{|k|}},
\end{equation*}
where we set
\begin{equation*}
  K(d, r, Q, \hbar, x) \defeq \exp\parentheses{ d\,r^{2} + \frac{2r}{\sqrt{\hbar}} \norm{ |Q|^{-1}x }_{1} }.
\end{equation*}
Furthermore, since $\abs{ w^{n} } \le \norm{ w }_{1}^{|n|}$ and $f^{k}_{n}(Q) = 0$ for $|n| \neq |k|$, we have, for a fixed $k \in \N_{0}^{d}$,
\begin{align*}
  \sum_{n \in \N_{0}^{d}} \abs{ f^{k}_{n}(Q)\, \frac{ p^{\hbar}_{k}\parentheses{ |Q|^{-1}(x - q) } }{k!}\, w^{n} }
  &\le \sum_{\substack{n \in \N_{0}^{d}\\ |n| = |k| }} \bigl| f^{k}_{n}(Q) \bigr|\, \frac{ \abs{ p^{\hbar}_{k}\parentheses{ |Q|^{-1}(x - q) } } }{k!}\, \norm{w}_{1}^{|k|} \\
  &\le \parentheses{ \sum_{n \in \N_{0}^{d}} \bigl| f^{k}_{n}(Q) \bigr| } \frac{ \abs{ p^{\hbar}_{k}\parentheses{ |Q|^{-1}(x - q) } } }{k!}\, \norm{w}_{1}^{|k|} \\
  &\le K(d, r, Q, \hbar, x - q)\, \parentheses{ \frac{ d\, \norm{ |Q|^{-1}\overline{Q} }_{\infty} \norm{w}_{1} }{r} }^{|k|}.
\end{align*}
Hence
\begin{equation*}
  \sum_{k \in \N_{0}^{d}} \sum_{n \in \N_{0}^{d}} \abs{ f^{k}_{n}(Q)\, \frac{ p^{\hbar}_{k}\parentheses{ |Q|^{-1}(x - q) } }{k!}\, w^{n} }
  \le K(d, r, Q, \hbar, x - q)\, \sum_{k \in \N_{0}^{d}} \parentheses{ \frac{ d\, \norm{ |Q|^{-1}\overline{Q} }_{\infty} \norm{w}_{1} }{r} }^{|k|}.
\end{equation*}
But then
\begin{align*}
  \sum_{k \in \N_{0}^{d}} \parentheses{ \frac{ d\, \norm{ |Q|^{-1}\overline{Q} }_{\infty} \norm{w}_{1} }{r} }^{|k|}
  &= \sum_{\ell=0}^{\infty} \sum_{\substack{k \in \N_{0}^{d}\\ |k| = \ell }} \parentheses{ \frac{ d\, \norm{ |Q|^{-1}\overline{Q} }_{\infty} \norm{w}_{1} }{r} }^{|k|} \\
  &= \sum_{\ell=0}^{\infty} {\ell + d - 1 \choose d - 1} \parentheses{ \frac{ d\, \norm{ |Q|^{-1}\overline{Q} }_{\infty} \norm{w}_{1} }{r} }^{\ell},
\end{align*}
which converges for those $w \in \C^{n}$ that satisfy $\norm{w}_{1} < r/\parentheses{ d\, \norm{ |Q|^{-1}\overline{Q} }_{\infty} }$.
Therefore, we can change the order of the double summation in \eqref{eq:double_sum} to obtain, for $\norm{w}_{1} < r/\parentheses{ d\, \norm{ |Q|^{-1}\overline{Q} }_{\infty} }$,
\begin{align*}
  \sum_{n \in \N_{0}^{d}} \mathcal{P}^{\hbar}_{n}(S, z; x)\,\frac{w^{n}}{n!} 
  = \sum_{n \in \N_{0}^{d}} \sum_{k \in \N_{0}^{d}} f^{k}_{n}(Q)\, \frac{ p^{\hbar}_{k}\parentheses{ |Q|^{-1}(x - q) } }{k!}\, w^{n}.
\end{align*}
The result follows by taking the derivatives of both sides at $w = 0$.
\end{proof}

\section*{Acknowledgments}
I would like to thank George Hagedorn for very kindly sharing with me his ideas and works on the multi-dimensional generalization of the minimal uncertainty product and the generating function, and also for being a constant source of inspiration.
I would like to acknowledge that it was Christian Lubich who suggested me to contact George Hagedorn regarding the minimal uncertainty product problem during my first visit at Universit\"at T\"ubingen during the Summer 2013, which was partially supported by the AMS--Simons Travel Grant.

The initial inspiration for the work on the ladder operators for the Hagedorn wave packets came from Stephanie Troppmann's talk at the Summer School on Mathematical and Computational Methods in Quantum Dynamics at University of Wisconsin--Madison during the Summer 2013.
I also thank Caroline Lasser for the helpful discussion during the workshop ``Mathematical and Numerical Methods for Complex Quantum Systems'' in March 2014 at the University of Illinois at Chicago regarding their work~\cite{LaTr2014}.
I would also like to take this opportunity to thank Shi Jin and Christof Sparber, the organizers of the summer school and workshop, for the generous travel support through the NSF Research Network in Mathematical Sciences ``KI-Net.''

I would also like to thank Stefan Teufel for hosting my second visit at Universit\"at T\"ubingen as a Simons Visiting Professor during the Summer 2015, where a part of this work was done.

Last but not the least, I would like to thank Johannes Keller and Stephanie Troppmann for the helpful discussions on generating functions during the Oberwolfach workshop ``Mathematical Methods in Quantum Molecular Dynamics'' in June 2015, and also Caroline Lasser for hosting my visit at Technische Universit\"at M\"unchen as a Simons Visiting Professor in June 2015, where a part of this work was done.
These research stays in Summer 2015 were partially supported by the Simons Foundation and by the Mathematisches Forschungsinstitut Oberwolfach.

\appendix
\numberwithin{equation}{section}

\section{The Heisenberg--Weyl and Metaplectic Operators}
\label{sec:HW_and_Mp}
This appendix gives a brief review of the Heisenberg--Weyl and metaplectic operators.
The purpose is to make the paper self-contained as well as accessible to a broad audience.
Our main reference is \citet[Chapters~3 \& 7]{Go2011}; see also \citet{Go2006} and \citet[Chapter~4]{Fo1989}.

\subsection{The Heisenberg--Weyl Operator}
\label{ssec:HW_operator}
First recall that the {\em Heisenberg--Weyl operator} $\widehat{T}_{z}$ with the parameter $z = (q,p) \in T^{*}\R^{d}$ is the unitary operator on $L^{2}(\R^{d})$ defined as follows:
\begin{equation}
  \label{eq:That}
  \widehat{T}_{z}\colon L^{2}(\R^{d}) \to L^{2}(\R^{d});
  \qquad
  (\widehat{T}_{z}f)(x) \defeq e^{\frac{\rmi}{\hbar}p\cdot(x - q/2)}\,f(x - q).
\end{equation}
We oftentimes restrict the domain of definition of $\widehat{T}_{z}$ to the Schwartz space $\mathscr{S}(\R^{d})$ and see $\widehat{T}_{z}$ as the operator $\widehat{T}_{z}\colon\mathscr{S}(\R^{d}) \to \mathscr{S}(\R^{d})$.

One may think of $\widehat{T}_{z_{0}}$ with a fixed $z_{0} \in T^{*}\R^{d}$ as a quantization of the phase space translation
\begin{equation*}
  T_{z_{0}}\colon T^{*}\R^{d} \to T^{*}\R^{d};
  \qquad
  z \mapsto z - z_{0}.
\end{equation*}
In fact, straightforward computations show that the standard position and momentum operators $\hat{z} = (\hat{x},\hat{p})$ satisfy
\begin{equation*}
  \widehat{T}_{z_{0}} \hat{z}\,\widehat{T}_{z_{0}}^{*} = \hat{z} - z_{0}
\end{equation*}
for any $z_{0} \in T^{*}\R^{d}$.

\subsection{The Metaplectic Group $\Mp(2d,\R)$}
\label{ssec:Mp}

The metaplectic group $\Mp(2d,\R)$ is a subgroup of the group $\U(L^{2}(\R^{d}))$ of the unitary operators on $L^{2}(\R^{d})$, and is generated by the following three classes of unitary operators on $L^{2}(\R^{d})$.
First we define $\widehat{J}\colon \mathscr{S}(\R^{d}) \to \mathscr{S}(\R^{d})$ as follows: For any $\psi \in \mathscr{S}(\R^{d})$,
\begin{equation}
  \label{eq:Jhat}
  \widehat{J}\psi(x) \defeq \frac{1}{ (2\pi\hbar\,\rmi)^{d/2} } \int_{\R^{d}} e^{-\frac{\rmi}{\hbar}x \cdot \tilde{x}}\, \psi(\tilde{x})\,d\tilde{x},
\end{equation}
and hence $\widehat{J} = \rmi^{-d/2} \mathscr{F}_{\hbar}$ with $\mathscr{F}_{\hbar}$ being the semiclassical Fourier transform~\eqref{eq:mathcalF_hbar}, i.e.,
\begin{equation*}
  \mathscr{F}_{\hbar}\psi(x) = \frac{1}{ (2\pi\hbar)^{d/2} } \int_{\R^{d}} e^{-\frac{\rmi}{\hbar}x \cdot \tilde{x}}\, \psi(\tilde{x})\,d\tilde{x}.
\end{equation*}
Therefore, we may think of $\widehat{J}$ as an isomorphism from $\mathscr{S}(\R^{d})$ to itself with its inverse given by
\begin{equation*}
  \widehat{J}^{-1} \psi(x)
  = \rmi^{d/2} \mathscr{F}_{\hbar}^{-1} \psi(x)
  = \parentheses{ \frac{\rmi}{ 2\pi\hbar }}^{d/2} \int_{\R^{d}} e^{\frac{\rmi}{\hbar} x \cdot \tilde{x}}\, \psi(\tilde{x})\,d\tilde{x}.
\end{equation*}
Since $\widehat{J}$ is essentially the Fourier transform $\mathscr{F}_{\hbar}$, one can easily extend it to the unitary operator $\widehat{J} \in \U(L^{2}(\R^{d}))$ and so $\widehat{J}^{*} = \widehat{J}^{-1}$.
Secondly, we define, for any $R \in \Sym(d,\R)$ (meaning $R$ is a $d \times d$ real symmetric matrix), $\widehat{V}_{R} \in \mathsf{U}(L^{2}(\R^{d}))$ as follows:
\begin{equation}
  \label{eq:Vhat}
  \widehat{V}_{R} \psi(x) \defeq e^{\frac{\rmi}{2\hbar}x^{T} R x} \psi(x).
\end{equation}
It is clearly a unitary operator on $L^{2}(\R^{d})$ with its inverse given by
\begin{equation*}
  (\widehat{V}_{R})^{-1} = (\widehat{V}_{R})^{*} = \widehat{V}_{-R}.
\end{equation*}
Lastly, for any $L \in \GL(d,\R)$, we define $\widehat{M}_{L}^{m} \in \mathsf{U}(L^{2}(\R^{d}))$ as follows:
\begin{equation}
  \label{eq:Mhat}
  \widehat{M}_{L}^{m} \psi(x) \defeq \rmi^{m} \sqrt{ |\det L| }\, \psi(L x),
\end{equation}
where the index $m \in \mathbb{Z}$ is defined by
\begin{equation*}
  m \pi \equiv \arg(\det L) \pmod{2\pi}.
\end{equation*}
This implies that there are two versions of $\widehat{M}_{L}^{m}$ that differ by the sign; see Remark~\ref{rem:Maslov} below.
Its inverse is given by
\begin{equation*}
  (\widehat{M}_{L}^{m})^{-1} = (\widehat{M}_{L}^{m})^{*} = \widehat{M}_{L^{-1}}^{-m}.
\end{equation*}
Alternatively, we may also write
\begin{equation}
  \label{eq:Mhat2}
  \widehat{M}_{L} \psi(x) \defeq (\det L)^{1/2}\, \psi(L x),
\end{equation}
where we incorporated the term $\rmi^{m}$ into the square root, and is taken care of by the branch chosen to define the square root.

Since the three classes of operators $\widehat{J}$, $\widehat{V}_{R}$, and $\widehat{M}_{L}^{m}$ are all elements of the group $\U(L^{2}(\R^{d}))$, one may consider the subgroup of $\U(L^{2}(\R^{d}))$ generated by these elements.
The {\em metaplectic group} $\Mp(2d,\R)$ is precisely this subgroup of $\U(L^{2}(\R^{d}))$, i.e., any element in $\Mp(2d,\R)$ is written as a composition of the above three classes of operators.

We may then construct (see \citet[Chapters~3]{Go2011} and \citet[Chapter~4]{Fo1989} for details) the homomorphism $\pi_{\Mp}\colon \Mp(2d,\R) \to \Sp(2d,\R)$ such that the generators $\widehat{J}$, $\widehat{V}_{R}$, and $\widehat{M}_{L}^{m}$ can be related to the following generators (see \citet[Corollary~63]{Go2011}) of the symplectic group $\Sp(2d,\R)$ as follows:
\begin{equation}
  \label{eq:pi_Mp_of_generators}
  \begin{array}{c}
    \pi_{\Mp}\parentheses{ \widehat{J} } =
    J =
    \begin{bmatrix}
      0 & I_{d} \\
      -I_{d} & 0
    \end{bmatrix},
           \quad
           \pi_{\Mp}\parentheses{ \widehat{V}_{R} } =
           V_{R} \defeq
           \begin{bmatrix}
             I_{d} & 0 \\
             R & I_{d}
           \end{bmatrix},
    \medskip\\
    \pi_{\Mp}\parentheses{ \widehat{M}_{L}^{m} } = 
    M_{L} \defeq
    \begin{bmatrix}
      L^{-1} & 0 \\
      0 & L^{T}
    \end{bmatrix}.
  \end{array}
\end{equation}
One can also show that $\ker\pi_{\Mp} = \{ \pm \id_{L^{2}(\R^{d})} \}$ and hence $\pi_{\Mp}\colon \Mp(2d,\R) \to \Sp(2d,\R)$ is a double cover.

In general, it is not straightforward to construct a concrete form of $\widehat{S} \in \Mp(2d,\R)$ for a given $S \in \Sp(2d,\R)$ such that $\pi_{\Mp}(\widehat{S}) = S$.
However, this can be done with a particular class of elements of $\Sp(2d,\R)$.
Specifically, let us define the set of {\em free symplectic matrices} as
\begin{equation}
  \label{eq:FSp}
  \mathsf{FSp}(2d,\R) \defeq
  \setdef{
    \begin{bmatrix}
      A & B \\
      C & D
    \end{bmatrix}
    \in \Sp(2d,\R)
  }{ A, B, C, D \in \Mat_{d}(\R),\, \det B \neq 0 }.
\end{equation}
Note that $\mathsf{FSp}(2d,\R)$ is {\em not} a subgroup of $\Sp(2d,\R)$ but just a subset of $\Sp(2d,\R)$.
One may then associate those classical linear canonical/symplectic transformations defined by elements of $\mathsf{FSp}(2d,\R)$ with the corresponding metaplectic operators in an explicit manner as follows:
For any free symplectic matrix $S =
\begin{tbmatrix}
  A & B \smallskip\\
  C & D
\end{tbmatrix} \in \mathsf{FSp}(2d,\R)$, one may define the corresponding quadratic function $W_{S}\colon \R^{d} \times \R^{d} \to \R$ by
\begin{equation*}
  W_{S}(\tilde{x}, x) \defeq \frac{1}{2} \tilde{x}^{T} B^{-1}A \tilde{x} - \tilde{x}^{T}B^{-1}x + \frac{1}{2} x^{T} D B^{-1} x.
\end{equation*}
This is the generating function for the canonical/symplectic transformation $\tilde{z} \defeq (\tilde{q}, \tilde{p}) \mapsto z = (q, p)$ defined by
\begin{equation}
  \label{eq:tildez_to_z}
  z = S \tilde{z}
  \quad \text{or} \quad
  \begin{bmatrix}
    q \\
    p
  \end{bmatrix}
  =
  \begin{bmatrix}
    A & B \\
    C & D
  \end{bmatrix}
  \begin{bmatrix}
    \tilde{q} \\
    \tilde{p}
  \end{bmatrix}
\end{equation}
in the sense that \eqref{eq:tildez_to_z} is equivalent to
\begin{equation*}
  \tilde{p} = -D_{1}W(\tilde{q}, q),
  \qquad
  p = D_{2}W(\tilde{q}, q),
\end{equation*}
where $D_{1}$ and $D_{2}$ stand for the partial derivatives with respect to the first and second variables, respectively.
Then we define the corresponding operator $\widehat{S}^{m}$ on $\mathscr{S}(\R^{d})$ as follows:
\begin{equation}
  \label{eq:quad_Fourier_transform}
  \widehat{S}^{m} \psi(x) \defeq \frac{\rmi^{m}}{\sqrt{ (2\pi\hbar\,\rmi)^{d}\, |\det(B)| }} \int_{\R^{d}} e^{\frac{\rmi}{\hbar}W_{S}(\tilde{x},x)}\, \psi(\tilde{x})\,d\tilde{x},
\end{equation}
where $m \in \mathbb{Z}$ is defined by
\begin{equation*}
  m \pi \equiv \arg(\det(B)) \pmod{2\pi}.
\end{equation*}
It is straightforward to check that
\begin{equation*}
  \widehat{S}^{m} = \widehat{V}_{D B^{-1}}\, \widehat{M}_{B^{-1}}^{m}\, \widehat{J}\, \widehat{V}_{B^{-1} A},
\end{equation*}
and hence $\widehat{S}^{m}$ is also an element in $\Mp(2d,\R)$, and also 
\begin{equation*}
  \pi_{\Mp}\parentheses{ \widehat{S}^{m} } = V_{D B^{-1}}\, M_{B^{-1}}\, J\, V_{B^{-1} A} =
  \begin{bmatrix}
    A & B \\
    C & D
  \end{bmatrix}
  = S.
\end{equation*}
Alternatively, we may write
\begin{equation}
  \label{eq:quad_Fourier_transform2}
  \widehat{S} \psi(x) \defeq \frac{(\det B)^{-1/2}}{ (2\pi\hbar\,\rmi)^{d/2} } \int_{\R^{d}} e^{\frac{\rmi}{\hbar}W_{S}(\tilde{x},x)}\, \psi(\tilde{x})\,d\tilde{x},
\end{equation}
where, as is the case with $\widehat{M}_{L}$, the sign due to the term $\rmi^{m}$ is determined by the branch chosen to define the square root in the factor $(\det B)^{-1/2}$.
Then we have
\begin{equation*}
  \widehat{S} = \widehat{V}_{D B^{-1}}\, \widehat{M}_{B^{-1}}\, \widehat{J}\, \widehat{V}_{B^{-1} A},
\end{equation*}
with appropriate choices of branches for $\widehat{S}$ and $\widehat{M}_{B^{-1}}$.

\begin{remark}
  \label{rem:Maslov}
  One can see that the above index $m$ is essentially in $\mathbb{Z}/4\mathbb{Z}$ as follows.
  If $\det B > 0$ then $m$ must be even, i.e., $m = 2l$ with $l \in \mathbb{Z}$ and so $\rmi^{m} = (-1)^{l}$; hence the sign of $\widehat{S}^{m} = \widehat{S}^{2l}$ depends on the parity of $l$:
  If $l$ is even, i.e., $l = 2k$ with $k \in \mathbb{Z}$, then $\widehat{S}^{m} = \widehat{S}^{4k}$ is the same operator for any $k \in \mathbb{Z}$, and if $l$ is odd, i.e., $l = 2k+1$ with $k \in \mathbb{Z}$, then $\widehat{S}^{m} = \widehat{S}^{4k+2}$ is the same for any $k \in \mathbb{Z}$ as well, and these two versions differ only by the sign, i.e., $\widehat{S}^{4k+2} = -\widehat{S}^{4k}$.
  Likewise, if $\det B < 0$ then $m = 2l + 1$ with $l \in \mathbb{Z}$ and so $\rmi^{m} = (-1)^{l}\,\rmi$, and thus the sign of $\widehat{S}^{m} = \widehat{S}^{2l+1}$ again depends on the parity of $l$:
  With $l = 2k$ and $k \in \mathbb{Z}$, $\widehat{S}^{m} = \widehat{S}^{4k+1}$ is the same for any $k \in \mathbb{Z}$ and with $l = 2k + 1$ and $k \in \mathbb{Z}$, the same goes with $\widehat{S}^{m} = \widehat{S}^{4k+3}$, and these two differ only by the sign, i.e., $\widehat{S}^{4k+3} = -\widehat{S}^{4k+1}$.
  That is, given any element $S \in \FSp(2d,\R)$, there exist {\em two} elements written as $\widehat{S}^{m}$.
  The same goes with the above definition~\eqref{eq:Mhat} of $\widehat{M}_{L}^{m}$.
\end{remark}

\begin{remark}
  \label{rem:abstractness_of_Mp}
  Unfortunately, not all the elements of $\Mp(2d,\R)$ are written in the form \eqref{eq:quad_Fourier_transform} or \eqref{eq:quad_Fourier_transform2}.
  However, one can show (see, e.g., \cite[Proposition~110]{Go2011}) that any element $\widehat{S} \in \Mp(2d,\R)$ may be written as the composition of two operators of the form \eqref{eq:quad_Fourier_transform2} (or \eqref{eq:quad_Fourier_transform}), i.e., $\widehat{S} = \widehat{S}_{1} \widehat{S}_{2}$ with those elements $S_{1}, S_{2} \in \mathsf{FSp}(2d,\R)$ such that $S = \pi_{\Mp}(\widehat{S}) = S_{1} S_{2}$, although this factorization is not unique.
\end{remark}

The integral expression~\eqref{eq:quad_Fourier_transform2} suggests that that the metaplectic operators $\widehat{S} \in \Mp(2d,\R)$ are, in a sense, a quantization of the linear symplectic transformation $z \mapsto S z$ on the phase space $T^{*}\R^{d}$ defined by the matrix $S \in \Sp(2d,\R)$.
This can be also illustrated by the following fact:
Taking the conjugation of the Heisenberg--Weyl operator~\eqref{eq:That} by a metaplectic operator $\widehat{S} \in \Mp(2d,\R)$ corresponding to $S \in \Sp(2d,\R)$, one obtains (see, e.g., \citet[Theorem~128 on p.~95]{Go2011})
\begin{equation}
  \label{eq:T-S_covariance}
  \widehat{S}\, \widehat{T}_{z}\, \widehat{S}^{*} = \widehat{T}_{S z}.
\end{equation}
Such a property is called {\em symplectic covariance}~\cite{Go2011}, and is very useful in calculations involving the Heisenberg--Weyl and metaplectic operators as illustrated in the main body of the paper; see, e.g., the proofs of Proposition~\ref{prop:A_covariance} and Corollary~\ref{cor:Hagedorn-Fourier}.

\section{The Hermite Functions and Hermite Polynomials}
\label{sec:Hermite}
This appendix is a summary of some facts on the Hermite functions and Hermite polynomials.
The purpose is mainly to set up our notation to avoid confusion due to a few different versions of definitions as well as to collect those results that are relevant to us.

\subsection{The Hermite Functions and Hermite Polynomials}
Let us start with the one-dimensional case.
Let $\tilde{\psi}_{n}$ be the $n$-th {\em Hermite function} with $n \in \N_{0} \defeq \N \cup \{0\}$, i.e., we have, for $x \in \R$,
\begin{equation*}
  \tilde{\psi}_{n}(x) \defeq \frac{\tilde{p}_{n}(x)}{\sqrt{2^{n}n!}\,\pi^{1/4}}\,\exp\parentheses{-x^{2}/2},
\end{equation*}
where $\tilde{p}_{n}$ is the $n$-th {\em Hermite polynomial}, i.e., $\tilde{p}_{0}(x) = 1$, $\tilde{p}_{1}(x) = 2x$, $\tilde{p}_{2}(x) = 4x^{2} - 2$, and so on.
Specifically, for $n = 0$, we have
\begin{equation*}
  \tilde{\psi}_{0}(x) = \frac{1}{\pi^{1/4}}\,\exp\parentheses{-x^{2}/2},
\end{equation*}
and so
\begin{equation*}
  \tilde{\psi}_{n}(x) = \frac{\tilde{p}_{n}(x)}{\tilde{c}_{n}}\,\tilde{\psi}_{0}(x)
\end{equation*}
with
\begin{equation*}
  \tilde{c}_{n} \defeq \sqrt{2^{n}n!}.
\end{equation*}

It is straightforward to generalize them to $d$-dimensions with $d \in \N$.
Let $n = (n_{1}, \dots, n_{d}) \in \N_{0}^{d}$ be a multi-index and $x = (x_{1}, \dots, x_{d}) \in \R^{d}$.
We define the {\em Hermite function} with the multi-index $n \in \N_{0}^{d}$ as
\begin{equation*}
  \psi_{n}(x) \defeq \prod_{j=1}^{d} \tilde{\psi}_{n_{j}}(x_{j})
  = \frac{p_{n}(x)}{c_{n}\,\pi^{d/4}}\,e^{-x^{2}/2}
  = \frac{p_{n}(x)}{c_{n}}\,\psi_{0}(x),
\end{equation*}
where
\begin{equation}
  \label{eq:c_n}
  c_{n} \defeq \tilde{c}_{n_{1}} \dots \tilde{c}_{n_{d}} = \sqrt{2^{|n|}n!}
\end{equation}
with $n! \defeq n_{1}! \dots n_{d}!$ and $|n| = n_{1} + \dots + n_{d}$, 
and $p_{n}$ is the {\em Hermite polynomial} with the multi-index $n \in \N_{0}^{d}$ defined as
\begin{equation*}
  p_{n}(x) \defeq \prod_{j=1}^{d} \tilde{p}_{n_{j}}(x_{j}),
\end{equation*}
and specifically, for $n = 0$, we have the Gaussian
\begin{equation*}
  \psi_{0}(x) = \frac{1}{\pi^{d/4}}\,\exp\parentheses{-x^{2}/2}.
\end{equation*}

Using the semiclassical scaling $x \to x/\sqrt{\hbar}$, we have the {\em semiclassically scaled Hermite functions}, i.e., for any $n \in \N_{0}^{d}$,
\begin{equation*}
  \psi^{\hbar}_{n}(x) \defeq \frac{1}{\hbar^{d/4}}\psi_{n}(x/\sqrt{\hbar})
  = \frac{p^{\hbar}_{n}(x)}{c_{n} (\pi\hbar)^{d/4}}\,e^{-\frac{x^{2}}{2\hbar}}
  = \frac{p^{\hbar}_{n}(x)}{c_{n}}\,\psi^{\hbar}_{0}(x),
\end{equation*}
where we defined the {\em semiclassically scaled Hermite polynomials}
\begin{equation}
  \label{eq:p^hbar_n}
 p^{\hbar}_{n}(x) \defeq p_{n}(x/\sqrt{\hbar})
\end{equation}
and particularly
\begin{equation*}
  \psi^{\hbar}_{0}(x) = \frac{1}{(\pi\hbar)^{d/4}}\,\exp\parentheses{-\frac{x^{2}}{2\hbar}}.
\end{equation*}
With the ladder operators defined by
\begin{equation*}
  \hat{a} \defeq \frac{1}{\sqrt{2\hbar}} (\hat{x} + \rmi\,\hat{p}),
  \qquad
  \hat{a}^{*} \defeq \frac{1}{\sqrt{2\hbar}} (\hat{x} - \rmi\,\hat{p}),
\end{equation*}
one sees that the Gaussian $\psi^{\hbar}_{0}$ is the ground state in the sense that $\hat{a}\,\psi_{0} = 0$, and also that, for any multi-index $n \in \N_{0}^{d}$ and $j \in \{1, \dots, d\}$,
\begin{equation}
  \label{eq:psi_n-ladders}
  \psi^{\hbar}_{n - e_{j}} = \frac{1}{\sqrt{n_{j}}}\, \hat{a}_{j} \psi^{\hbar}_{n},
  \qquad
  \psi^{\hbar}_{n + e_{j}} = \frac{1}{\sqrt{n_{j} + 1}}\, \hat{a}^{*}_{j} \psi^{\hbar}_{n},
\end{equation}
where $e_{j}$ is the unit vector in $\R^{d}$ whose $j$-th entry is 1, and $n_{j} \ge 1$ is assumed in the first equation.

\subsection{Generating Function}
Again, let us start with the one-dimensional case.
The generating function for the one-dimensional Hermite polynomials $\{ \tilde{p}_{n} \}_{n \in \N_{0}}$ is defined as
\begin{equation*}
  \tilde{\gamma}(w, x) \defeq \exp( 2w x  - w^{2} )
\end{equation*}
and satisfies
\begin{equation*}
  \tilde{\gamma}(w, x) = \sum_{n = 0}^{\infty} \tilde{p}_{n}(x)\,\frac{w^{n}}{n!}
\end{equation*}
for $x \in \R$ and $w \in \C$, or equivalently,
\begin{align*}
  \tilde{\Gamma}(w, x) &\defeq \psi_{0}(x)\,\tilde{\gamma}(w, x)
  = \frac{1}{\pi^{1/4}} \exp\parentheses{ -\frac{x^{2}}{2} + 2w x  - w^{2} } \\
  &= \sum_{n = 0}^{\infty} \tilde{p}_{n}(x)\,\tilde{\psi}_{0}(x)\,\frac{w^{n}}{n!}
  = \sum_{n = 0}^{\infty} \tilde{\psi}_{n}(x)\,\frac{\tilde{c}_{n}}{n!}\,w^{n}.
\end{align*}
The generating function $\tilde{\gamma}$ can be exploited along with the Cauchy integral formula to give the following estimate for the Hermite polynomials (see, e.g., \citet[Exercise 7.4]{Ar1997}):
For any $n \in \N_{0}$ and $r > 0$, we have
\begin{equation}
  \label{eq:tildep_n-estimate}
  \bigl| \tilde{p}_{n}(x) \bigr| \le \frac{n!}{r^{n}} \exp(r^{2} + 2r|x|).
\end{equation}
The multi-dimensional generating function is the following simple product of the one-dimensional generating functions:
\begin{equation*}
  \gamma(w, x) \defeq \prod_{j=1}^{d} \tilde{\gamma}(w_{j}, x_{j})
  = \exp( 2w^{T}x  - w^{2} ) = \sum_{n \in \N_{0}^{d}} p_{n}(x)\,\frac{w^{n}}{n!}
\end{equation*}
or
\begin{align*}
  \Gamma(w, x) &\defeq \prod_{j=1}^{d} \tilde{\Gamma}(w_{j}, x_{j})
                 = \psi_{0}(x)\,\gamma(w,x) \\
               &= \frac{1}{\pi^{d/4}} \exp\parentheses{ -\frac{x^{2}}{2} + 2w^{T}x  - w^{2} } \\
               &= \sum_{n \in \N_{0}^{d}} p_{n}(x)\,\psi_{0}(x)\,\frac{w^{n}}{n!}
                 = \sum_{n \in \N_{0}^{d}} \psi_{n}(x)\,\frac{c_{n}}{n!}\,w^{n},
\end{align*}
where $w = (w_{1}, \dots, w_{d}) \in \C^{d}$ and $w^{n}$ stands for $w_{1}^{n_{1}} \dots w_{d}^{n_{d}}$.
The above estimate~\eqref{eq:tildep_n-estimate} on the one-dimensional Hermite polynomials can be easily extended to the multi-dimensional Hermite polynomials:
For any $n \in \N_{0}^{d}$ and $r > 0$, we have
\begin{equation}
  \label{eq:p_n-estimate}
  | p_{n}(x) | \le \frac{n!}{r^{|n|}} \exp( d\,r^{2} + 2r\norm{x}_{1} ),
\end{equation}
where $\norm{x}_{1} \defeq \sum_{j=1}^{d} |x_{j}|$.
With the semiclassical scaling, we have the following generating function shown in \eqref{eq:gamma-Hermite}:
\begin{equation*}
  \gamma^{\hbar}(w, x)
  \defeq \gamma(w, x/\sqrt{\hbar})
  = \exp\parentheses{ \frac{2}{\sqrt{\hbar}}\,w^{T}x - w^{2} } = \sum_{n \in \N_{0}^{d}} p^{\hbar}_{n}(x)\,\frac{w^{n}}{n!},
\end{equation*}
where $p^{\hbar}_{n}(x) \defeq p_{n}(x/\sqrt{\hbar})$, and so 
\begin{align*}
  \Gamma^{\hbar}(w,x)
  &\defeq \frac{1}{\hbar^{d/4}}\,\Gamma^{\hbar}(w,x/\sqrt{\hbar}) \\
  &= \frac{1}{(\pi\hbar)^{d/4}}\,\exp\parentheses{ -\frac{x^{2}}{2\hbar} + \frac{2}{\sqrt{\hbar}}\,w^{T}x - w^{2} } \\
  &= \sum_{n \in \N_{0}^{d}} p^{\hbar}_{n}(x)\,\psi_{0}^{\hbar}(x)\,\frac{w^{n}}{n!}
  = \sum_{n \in \N_{0}^{d}} \psi^{\hbar}_{n}(x)\,\frac{c_{n}}{n!}\,w^{n},
\end{align*}
which is \eqref{eq:Gamma-Hermite}.

\bibliography{Hagedorn-Hermite}
\bibliographystyle{plainnat}

\end{document}